\documentclass[12pt, draftclsnofoot, onecolumn]{IEEEtran}
\makeatletter
\newcommand*\titleheader[1]{\gdef\@titleheader{#1}}
\AtBeginDocument{%
  \let\st@red@title\@title
  \def\@title{%
    \bgroup\normalfont\large\centering\@titleheader\par\egroup
    \vskip1.5em\st@red@title}
}
\makeatother

\usepackage[T1]{fontenc}%
\ifCLASSINFOpdf
  \usepackage[pdftex]{graphicx}
\else
  \usepackage[dvips]{graphicx}

\fi

\usepackage{array}

\ifCLASSOPTIONcompsoc
  \usepackage[caption=false,font=normalsize,labelfont=sf,textfont=sf]{subfig}
\else
  \usepackage[caption=false,font=footnotesize]{subfig}
\fi

\usepackage{amsmath}
\usepackage{amsthm}
\usepackage{amssymb} 
\usepackage{bbm}
\usepackage{cite}

\usepackage{url}

\usepackage[dvipsnames]{xcolor}
\usepackage{comment}

\usepackage{algorithm}
\usepackage{algorithmic}
\usepackage{enumitem}

\graphicspath{{figures/}}
\newtheorem{theorem}{Theorem}
\newtheorem{lemma}{Lemma}
\newtheorem{claim}{Claim}
\newtheorem{remark}{Remark}

\newtheorem{definition}{Definition}

\DeclareMathOperator*{\argmax}{arg\,max}

\allowdisplaybreaks

\newcommand{\MYfooter}{\smash{\scriptsize
\hfil\parbox[t][\height][t]{\textwidth}{\centering}\hfil\hbox{}}}

\makeatletter
\def\ps@headings{%
\def\@oddhead{\mbox{Submitted to IEEE Transactions on Wireless Communication}\scriptsize\rightmark \hfil \thepage}
\def\@evenhead{\scriptsize\thepage \hfil \leftmark\mbox{}}
\def\@oddfoot{\MYfooter}%
\def\@evenfoot{\MYfooter}}
\def\ps@IEEEtitlepagestyle{%
\def\@oddhead{\mbox{Submitted to IEEE Transactions on Wireless Communication }\scriptsize\rightmark \hfil \thepage}%
\def\@evenhead{\scriptsize\thepage \hfil \leftmark\mbox{}}%
\def\@oddfoot{\MYfooter}%
\def\@evenfoot{\MYfooter}}

\makeatother

\begin{document}
%

\title{Relay Selection, Scheduling and Power Control in Wireless Powered Cooperative Communication Networks}

%
%
%

\author{Aysun~Gurur~Onalan, Elif~Dilek~Salik, \textit{Student Member}, Sinem~Coleri, \textit{Senior Member} \thanks{A. G. Onalan, E. D. Salik and S. Coleri are with the department of Electrical and Electronics Engineering, Koc University, Istanbul, Turkey, e-mail:~\texttt{\{aonalan17, esalik, scoleri\}@ku.edu.tr}. This work is supported by Scientific and Technological Research Council of Turkey Grant $\#$117E241. A preliminary version of the paper considering only part of the scheduling and power control has appeared in \cite{Salik19}.} }
\maketitle

\begin{abstract}

Relay nodes are used to improve the throughput, delay and reliability performance of energy harvesting networks by assisting both energy and information transfer between information nodes and access point. Previous studies on radio frequency energy harvesting networks are limited to single source single/multiple relay networks. In this paper, a novel joint relay selection, scheduling and power control problem for multiple source multiple relay network is formulated with the objective of minimizing the total duration of wireless power and information transfer. The formulated problem is non-convex mixed-integer non-linear programming problem, and proven to be NP-hard. We first formulate a sub-problem on scheduling and power control for a given relay selection. We propose an efficient optimal algorithm based on  a bi-level optimization over power transfer time allocation. Then, for optimal relay selection, we present optimal exponential-time Branch-and-Bound (BB) based algorithm where the nodes are pruned with problem specific lower and upper bounds. We also provide two BB-based heuristic approaches limiting the number of branches generated from a BB-node, and a relay criterion based lower complexity heuristic algorithm. The performance of the proposed algorithms are demonstrated to outperform conventional harvest-then-cooperate approaches with up to $88\%$ lower schedule length for various network settings. 

\end{abstract}


\section{Introduction}

Radio frequency energy harvesting (RF-EH) has drawn significant attention as an alternative to fixed power supplies (e.g. batteries) \cite{Lu2015}. RF-EH has also broadened the research on energy harvesting network designs as the RF signals can carry both energy and information. In the literature, two main RF-EH network designs are proposed: Simultaneous Information and Power Transfer (SWIPT)  \cite{Perera2018}, and Wireless Powered Communication Networks (WPCN) \cite{Kang2015, sinem2019}. In SWIPT, information and energy are sent simultaneously from the access point (AP) to multiple network nodes. In WPCN, nodes harvest energy from signals broadcast by the AP in downlink (DL), and then send their message to the AP in uplink (UL) by exerting this harvested energy. 


Recently, relay nodes have been considered in RF-EH networks for further improvement of network performance in terms of efficiency and reliability \cite{Lu2015}. In conventional networks, relaying is a promising strategy to increase system throughput, network coverage and energy efficiency by assisting information transmission of  users  \cite{Chen2012,Nosratinia2004}.  In RF-EH networks, apart from information forwarding, relays can also broadcast more energy to the users \cite{Mishra2017, Ammar2018}. This requires balancing the trade-off in the usage of relays for energy and information transmission. 

The first widely studied relay based energy harvesting network is the three node model, consisting of a source, a relay and a destination. \cite{Nasir2015,Gu2015,Ju2015} consider an EH relay in SWIPT where the source and the destination have an unlimited power  source. In these studies, the EH relay uses part of the information signal sent by the source for energy harvesting; while the remaining part of the signal is used to forward information to the destination. In three node WPCN model, which is also referred as Wireless Powered Cooperative Communication Network (WPCCN), both relay and source need to harvest energy \cite{Chen2015,Gu2015a,Li2016}. In the DL, AP broadcasts RF signals for energy harvesting at the relay and source. Then, in the first half of the information transmission (IT) period, the source broadcasts the information to the relay, or both  the relay and AP. In the second half, the relay conveys the information to the AP. For both SWIPT and WPCCN, average throughput and outage probability of the network are derived under different conditions (e.g. amplify or decode and forward relaying, with/without direct link between source and destination). Numerical results demonstrate the positive effect of relays on the outage and throughput of RF-EH networks.


SWIPT is extended for relay selection in a network containing one source, one destination, and multiple EH relays \cite{Krikidis2014,Gu2018,Wang2018,Zhao2016,Sui2018,Nasir2018}. After a relay is selected, simultaneous energy and data transmission is performed as in the three node model. A straightforward approach to the relay selection problem is to choose the relay randomly inside a circular sector with a central angle in direction of the destination \cite{Krikidis2014}. The relay selection schemes are also designed based on the remaining energies in relay batteries (for the cases where EH relays are equipped with batteries) \cite{Gu2018,Wang2018}, or  channel state information (CSI) either between source-and-relay, relay-and-destination or both  \cite{Zhao2016,Sui2018,Nasir2018}. Outage probability and average throughput of the networks are analyzed as performance indicator. 


The relay selection problem for WPCCN is only addressed in \cite{Chen2015}. In this work, the network consists of an AP, an EH source, and multiple EH relays. The first part of a time block is allocated for power transfer from the AP to the source and relays. The remaining is further divided into two equal length slots for IT in which the source first broadcasts the information towards the AP and relays, and then the selected relay amplifies and forwards the signal to the AP.  The performance of the aforementioned CSI based relay selection criteria is evaluated in terms of outage probability and throughput for fixed transmission rate. Although the optimal fraction of the time block between EH and IT to maximize throughput is discussed through simulation results, durations of the IT slots were not optimized but assumed equal. However, this optimization can bring significant improvement, especially for time critical network applications, such as real-time surveillance and networked control systems \cite{Sadi2014,Hou2012}. Furthermore, optimization of power control and rate adaption were not studied.

RF-EH networks with multiple source-destination pairs require separate study for the relay selection problem since individually optimum relays do not guarantee the best relay selection for the entire network. To the best of our knowledge, the multiple source relay selection problem is considered only in \cite{Ding2014} for SWIPT  and has not been analyzed for WPCCN in the literature. In SWIPT, relay selection is performed by employing a game theoretic approach, where  the sources are competing for the help of relays. 


In this paper, we study multiple source multiple relay WPCCN where relays and sources harvest energy from the RF signals broadcast by an AP for information transmission. Our goal is to determine the optimal relay selection, power control, and scheduling for energy  harvesting and information  transmission. The major contributions of the paper are listed as follows.
\begin{itemize}
\item We consider multiple source multiple relay WPCCN; whereas the earlier work in the literature is limited to the single source single/multiple relay WPCCN. 

\item A novel optimization problem is formulated for joint relay selection, scheduling and power control. Besides the usual energy and data causality constraints, we include maximum transmission power constraints. The objective is to minimize the duration of the schedule for wireless power transfer and data transmission. The formulated problem is a non-convex mixed integer non-linear problem (MINLP).

\item The formulated problem is proved to be NP-hard, so, as a solution strategy, we propose a bottom-up approach starting from a simpler sub-problem of scheduling and power control, then, extending the problem with additional relay selection variables.


\item A scheduling and power control problem is formulated where the objective is to minimize the schedule length for the energy and data transmission and the constraints concern energy, data, and maximum transmit power requirements. An efficient optimum and a faster sub-optimal algorithm are proposed for the solution of the problem. The former is based on the bi-level transformation of the problem based on the convexity of the individual optimal IT lengths with respect to the EH length; whereas the latter computes the IT duration of the sources for a sub-optimal EH length. 


\item Extending scheduling and power control problem with relay selection variables requires the search over all possible source-relay combinations. We adapt Branch-and-Bound to the solution of our problem so that we propose an optimum and two heuristic algorithms. The optimal algorithm systematically searches for  relay-source combinations by following the nodes of a BB-tree. A node is pruned if it is not promising to provide better feasible solution than the one obtained at the previous nodes. For pruning decision, problem specific lower and upper bounds are developed. The heuristic approaches either limit the number of branches generated from each node or completely remove the branching to obtain solutions quickly. After any relay assignment is completed, there is no need to branch on other variables. The scheduling and power control problem is solved to obtain the schedule length. 

\item Although the BB-based heuristics shorten the time to reach a feasible solution, their run-time still rapidly increases with the number of the nodes of the WPCCN. For better time complexity, an improvement type heuristic approach is proposed. The heuristic starts with an initial solution determined by CSI-based relay selection criterion, and then, alters the selected relay for a source in each iteration until there is no improvement in schedule length. 

\item We illustrate the superiority of the proposed algorithms to previously proposed conventional harvest-then-cooperate protocol \cite{Chen2015} in terms of schedule length and run-time for various network sizes and transmit power levels.

\end{itemize}

The rest of the paper is organized as follows. Section~\ref{section:sysmod} describes the system model and assumptions used throughout the paper. In Section~\ref{section:probForm}, the joint relay selection, scheduling, and power control problem is formulated.  Section~\ref{section:subprob} introduces and solves the sub-problem on scheduling and power control. The BB based optimal and heuristic algorithms are proposed  in Sections~\ref{section:relsel} and \ref{section:sBB2}, respectively. Section~\ref{seciton:relCritHeur} provides a relay selection criterion based heuristic algorithm. Section~\ref{section:simRes} presents the simulation results. Finally, Section~\ref{section:conc} concludes the paper.



\section{System Model} \label{section:sysmod}



\begin{itemize}
\item  The WPCCN contains one AP, $N$ information sources, denoted by $S_i$, $i = 1,2,\ldots, N$, and $K$ decode-and-forward (DF) relays, denoted by $R_j$, $j=1,2,\ldots K$.  The AP, all the sources and relays are equipped with a single antenna and operate in the same frequency band. 

\item  The AP is assumed to have an unlimited power source. The transmission power of the AP is assumed to be constant and denoted by $P_{A}$. The sources and relays only have rechargeable batteries and no other embedded energy supplies. Therefore, they have to use the harvested energy from the RF signals broadcast by the AP for the information transmission.

\item  TDMA protocol is used as medium access control protocol. The time is divided into time blocks, which are further divided into 3 time slots for EH, IT of sources, and IT of relays, respectively.  In the first time slot, relays and information sources harvest energy from the AP during $\tau_0$ amount of time.  With the harvested energy, each source communicates directly to the AP, or information is transferred from the source to the relay and then from the relay to the AP. Accordingly, IT time slots are further divided into sub-slots corresponding to each source and relay.  Assume that  the $j^{th}$ relay is selected by the $i^{th}$ source for cooperation. Note that each source can cooperate with a single relay. Then, a sub-slot is allocated for the IT from source $S_i$ to relay $R_j$ with duration $\tau_{S_i}^{R_j}$ and another sub-slot is allocated for the IT from $R_j$ to $AP$ with duration $\tau_{R_j}^{AP}$ for $i = 1,2,\ldots, N$ and $j=0,1,2,\ldots K$, where $R_0$ refers to $AP$.  If the $i^{th}$ source prefers direct communication with the AP, information transfer is completed in the first IT block in $\tau_{S_i}^{R_0}$ amount of time and $\tau_{R_0}^{AP}$ is obviously zero.  Further, it is assumed that the harvested energy in each block is only used in that block and not stored for further use.

\item The AP maintains the synchronization of the network and informs the information sources and relays about the current schedule at the beginning of each time block.
\item Block fading channels are assumed, where the channel gains remain constant during each transmission block but change independently from one block to another.  Let $h_{X}^Y $ and  $g_{X}^Y $  be the channel gains between $X$ and $Y$ for DL and UL, respectively, where $X,Y \in \{AP, S_i, R_j\}$, $i = 1,2,\ldots, N$ and $j=1,2,\ldots K$ (e.g. $h_{AP}^{S_i}$ is the DL channel gain between $AP$ and $S_i$).  All the channel gains are assumed to be known by the AP at the beginning of each time block, similar to the previous work, e.g., \cite{Ju2014,Zhao2016,Kang2015}.
x
\item Each information source $S_i$ needs to send $D_{S_i}$ amount of data to the AP, for $i = 1,2,\ldots, N$. Moreover, the relays must convey all the information they receive from the sources.

\item  We assume that $P_{A}$  is large enough to ignore AWGN noise component during EH phase and all the sources and relays have constant harvesting efficiency. Thus, the energy harvested at the $i^{th} $ source and $j^{th}$ relay can be expressed as
\begin{equation}
	E_{S_i} = \zeta_i \tau_0  P_A h_{AP}^{S_i} \text{ and } E_{R_j} = \zeta_j \tau_0  P_A h_{AP}^{R_j},
\label{eq:EHatX}
\end{equation}
respectively, where $\zeta_i$ and $\zeta_j$ are the energy harvesting efficiency of $S_i$ and $R_j$, respectively, for $i=1,\ldots, N$ and  $j=1,\ldots, K$.

\item In IT phase, we consider continuous power model, which is a commonly used model in the literature \cite{Zhang2016,Sadi2014}, so the transmit power takes any value below a maximum level $P^{max}$. The transmit power is further limited by the amount of the harvested energy given in Eq.~(\ref{eq:EHatX}). The required energy for decoding the received message at relays is negligible compared to the energy consumption for IT \cite{Gu2015}. Thus, 
\begin{equation} 
P_{S_i}^{R_j} \leq \frac{E_{S_i}}{\tau_{S_i}^{R_j}} \text{ and } P_{R_j}^{AP} \leq \frac{E_{R_j}}{\tau_{R_j}^{AP}},
\label{eq:transPowSource} 
\end{equation}
where $P_{S_i}^{R_j}$ and $P_{R_j}^{AP}$ are the average transmit power of $S_i$ and $R_j$ during IT, respectively.




\item We use continuous transmission rate model and Shannon's channel capacity formula for AWGN channels to determine maximum achievable rate. We ignore interference as all sources and relays use separate time slots for IT. Hence, the instantaneous UL transmission rates $T_{S_i}^{R_j}$ from ${S_i}$ to ${R_j}$  and $T_{R_j}^{AP}$ from ${R_j}$ to $AP$ are given by
\begin{align}
T_{S_i}^{R_j} &= W \log_2 \left( 1 + P_{S_i}^{R_j}g_{S_i}^{R_j}/(WN_0) \right) \\
T_{R_j}^{AP} &= W \log_2 \left( 1 +  P_{R_j}^{AP}g_{R_j}^{AP}/(WN_0) \right)
\label{eq:rateS}
\end{align}

\end{itemize}

\section{Problem Formulation} \label{section:probForm}

A joint relay selection, time allocation and power control problem with the objective of total time minimization subject to demand, energy and power constraints is formulated as follows.

\begin{subequations} \label{eq:relSel}
\begin{align}
\min \qquad &  \tau_0 + \sum_{i=1}^N \sum_{j=0}^K \tau_{S_i}^{R_j} + \sum_{j=1}^K \tau_{R_i}^{AP} \label{eq:relSel:obj} \\
\text{s.t.}\qquad & \sum_{j=0}^K b_i^j = 1, \qquad i=1,\ldots,N, \label{eq:relSel:relSel}\\
									& P_{S_i}^{R_j}  \leq P^{max} b_i^j, \quad i=1,\ldots,N, \; j=0,\ldots,K, \label{eq:relSel:maxpow1} \\
									& P_{R_j}^{AP}  \leq  \min \left\{P^{max},P^{max} \sum_{i=1}^N  b_i^j \right\}, \quad j=0,\ldots,K, \label{eq:relSel:maxpow2} \\
									& P_{S_i}^{R_j} \tau_{S_i}^{R_j} \leq \zeta_i P_{A} h_{AP}^{S_i} \tau_0, \qquad i=1,\ldots,N, \quad j=0,\ldots,K, \label{eq:relSel:energy1}\\
									& P_{R_j}^{AP} \tau_{R_j}^{AP} \leq \zeta_j P_{A} h_{AP}^{R_j} \tau_0, \qquad j=1,\ldots,K, \label{eq:relSel:energy2}\\
									& \tau_{S_i}^{R_j} W \log_2 \left( 1 + \frac{P_{S_i}^{R_j}g_{S_i}^{R_j}}{WN_0} \right)\geq D_{S_i}b_i^j, \quad i=1,\ldots,N, \; j=0,\ldots,K,  \label{eq:relSel:demand1}\\
									& \tau_{R_j}^{AP} W \log_2 \left( 1 + \frac{P_{R_j}^{AP}g_{R_j}^{AP}}{WN_0} \right) \geq \sum_{i=1}^N  D_{S_i}b_i^j, \quad j=0,\ldots,K, \label{eq:relSel:demand2}\\									
									&  \tau_0,\tau_{R_j}^{AP}, P_{R_j}^{AP},\tau_{S_i}^{R_j}, P_{S_i}^{R_j} \geq 0, \quad i=1,\ldots,N, \; j=1,\ldots,K \label{eq:relSel:nonneg}\\
									& b_i^j \in \{0,1\},  \quad i=1,\ldots,N, \; j=1,\ldots,K.		\label{eq:relSel:integrality}	
\end{align}
\end{subequations}

The variables of the optimization problem are $\tau_0$, EH duration; $\tau_{S_i}^{R_j}$, IT duration from the $i^{th}$ source to the $j^{th}$ relay;  $\tau_{R_j}^{AP}$, IT duration from the $j^{th}$ relay to the AP; $P_{S_i}^{R_j}$, transmit power of  the $i^{th}$ source when it transmits to the $j^{th}$ relay; $P_{R_j}^{AP}$, transmit power of the $j^{th}$ relay when it transmits to the AP; and $b_i^j$, the relay selection parameter taking value 1 if the $j^{th}$ relay is selected for the $i^{th}$ source and 0 otherwise.

The objective of the optimization problem is to minimize the duration of the schedule for energy harvesting and data transmission. Eq.~(\ref{eq:relSel:relSel}) guarantees that one and only one relay is selected for each source. Note that a relay can be chosen for more than one source. Eqs.~(\ref{eq:relSel:maxpow1}) and (\ref{eq:relSel:maxpow2}) determine the maximum allowable transmit power of  the sources and relays, respectively. If $R_j$ is not selected for $S_i$, $P_{S_i}^{R_j}$ is forced to be zero in Eq.~(\ref{eq:relSel:maxpow1}).  If $R_j$ is not selected for any source, then $P_{R_j}^{AP}$ is set to zero  by Eq.~(\ref{eq:relSel:maxpow2}).   Eqs.~(\ref{eq:relSel:energy1}) and (\ref{eq:relSel:energy2}) represent the limit on the transmit power of sources and relays imposed by the harvested energy amount. Demand constraint  Eq.~(\ref{eq:relSel:demand1}) requires that each source $S_i$ conveys a certain amount $D_{S_i}$ of data to the AP with/without  the help of a relay. Relays need to convey all data they gather from users. Accordingly, demand requirement of a relay is given in Eq.~(\ref{eq:relSel:demand2}). Lastly, Eq.~(\ref{eq:relSel:nonneg}) and (\ref{eq:relSel:integrality}) represent non-negativity and integrality constraints, respectively.

The described problem is a non-convex MINLP as all equations except Eqs.~(\ref{eq:relSel:relSel})-(\ref{eq:relSel:maxpow2}) are non-linear and Eqs. (\ref{eq:relSel:energy1})-(\ref{eq:relSel:demand2}) are non-convex. Next, we show the NP-hardness of the problem.

\subsection{NP-Hardness}
\begin{theorem}  The joint relay selection, time allocation and power control problem in Eq.~(\ref{eq:relSel}) is NP-hard. \end{theorem}

\begin{proof} 

We reduce the NP-hard minimum Makespan Scheduling Problem (MSP) on identical machines to an instance of Problem~(\ref{eq:relSel}). MSP aims to find an assignment of the $N$ jobs with processing times $t_i, i=1,\ldots,N,$ to $K$ identical machines such that the makespan, which is the time elapsed until all jobs completed, is minimized. 

Let us define a problem instance where UL channels gains are equal to each other as DL gains, i.e., $g_{X}^Y =g$ and $h_{X}^Y =h$, where $X,Y \in \{AP, S_i, R_j\}, \forall  i=1,\ldots,N, \; j=0,\ldots,K$. In this condition, the relays become identical. Any relay brings same amount of improvement in IT duration of a source node, which removes the effect of the relay selection on IT durations. As the channel gains are same, the required energy for $D_{S_i}$ amount of data transmission from $S_i$ to $R_j$ is same for any $j$, i.e., not affected by the relay assignment. However, a relay assisting many sources necessitates more energy for IT which results in longer EH duration. For the minimum EH duration, the data amount forwarded by the relays should be balanced. Consider that $D_{S_i}$ as processing time for $S_i, \forall i,j$ as IT lengths and required energy are directly related to the data amount. Hence, the problem instance turns into an MSP which asks for an assignment of $N$ sources with different data demands $D_i$ to $K$ identical relays such that the maximum of total data amount forwarded by a relay is minimized. Since MSP is NP-hard and  reduced to an instance of the problem, Problem~(\ref{eq:relSel}) is NP-hard. \end{proof}

\subsection{Solution Strategy}
For the solution of Problem~(\ref{eq:relSel}), we first formulate a sub-problem on scheduling and power control for WPCN with multiple source-destination pairs and an AP as energy source, which is the reduced form of the original problem for a given relay selection. A preliminary version of this problem and the solution algorithms has previously appeared in \cite{Salik19}. 
Then, the best relay selection is searched either by BB based techniques or an improvement type heuristic approach.

\section{Scheduling and Power Control Problem} \label{section:subprob}

This section discusses the  scheduling  and  power  control  problem  for WPCN  with  multiple  source-destination  pairs. Each source has  a  destination, which  is either the  selected  relay targeting the AP  or  the  AP if no  relay  is  chosen  for  that  source.  Since relays have to convey the information gathered from the sources to the AP, they act like sources. Therefore, in this section, the  term source covers both  the sources and the selected  relays, the total number of which is denoted by $N$. 

During $\tau_0$ amount of time, sources harvest energy from the  signals broadcast by the AP in DL. Then, the $i$-th source transmits $D_{S_i}$ amount of data to the corresponding destination during $\tau_{S_i}$ amount of time in UL, for $ i=1,2,\ldots,N$. Accordingly, the scheduling and power control problem is given as follows:
\begin{subequations} \label{eq:problm1APNS}
\begin{align}
\min \qquad &  \tau_0 + \sum_{i=1}^{N} \tau_{S_i} \\
\text{s.t.}\qquad & P_{S_i} \tau_{S_i} \leq \zeta_i P_{A} h_{S_i} \tau_0, \; \quad  i=1,2,\ldots,N \label{eq:energy1APNS}\\
									& \tau_{S_i} W \log_2 \left( 1 + \frac{P_{S_i}g_{S_i}}{WN_0} \right) \geq D_{S_i} , \quad i=1,2,\ldots,N \label{eq:data1APNS}\\
									& P_{S_i} \leq P^{max}, \qquad  i=1,2,\ldots,N \label{eq:pmax1APNS}  \\
									&\tau_0, \tau_{S_i}, P_{S_i}  \geq 0, \qquad i=1,2,\ldots,N \label{eq:nonneg1APNS}
\end{align}
\end{subequations}
where $h_{S_i}$ is the DL channel gain between $S_i$ and AP, $g_{S_i}$ is the UL channel gain between $S_i$ and its destination, $P_{S_i}$ is the transmit power of $S_i$, for $i=1,\ldots, N$.  


Problem~(\ref{eq:problm1APNS}) is a non-convex non-linear optimization problem due to the non-convexity of Eqs.~(\ref{eq:energy1APNS}) and (\ref{eq:data1APNS}). We first give the optimal solution of the single source scenario. Then, we extend our solution strategy to the multiple source scenario by modeling the system as multiple single source networks and designing a bi-section search for a mutual EH duration of such networks. Moreover, single source solutions will be used in the sub-optimal algorithm, where optimum EH lengths of each individual source are computed and the maximum of them is set as mutual EH duration. 

\subsection{Single Source}
In this subsection, we consider a single source WPCN, i.e., $N=1$. We first analyze Problem~(\ref{eq:problm1APNS}) without constraint (\ref{eq:pmax1APNS}), and then provide the optimal solution by additionally considering constraint (\ref{eq:pmax1APNS}) in  Theorem~\ref{thm:singleSource}. 

For Problem~(\ref{eq:problm1APNS}) without constraint (\ref{eq:pmax1APNS}), when there is a single source in the network, EH duration $\tau_{0}$ should be adjusted just to meet $S_1$'s power requirement, i.e., Eq.~(\ref{eq:energy1APNS}) holds with equality. Therefore, we can combine Eqs.~(\ref{eq:energy1APNS}) and (\ref{eq:data1APNS}) as

\begin{equation}
\tau_{S_1} W \log_2 \left( 1 + \frac{\zeta_1 P_{A} h_{S_1} \tau_0g_{S_1}}{\tau_{S_1}WN_0} \right) \geq D_{S_1},
\label{eq:combEnergyAndData}
\end{equation}
which represents a convex set. Before Theorem~\ref{thm:singleSource}, we present the following lemmas and definitions that are used in the proof of the theorem.

\begin{lemma}\label{lemma:f}
Let us rewrite Eq.~(\ref{eq:combEnergyAndData}) such that $\tau_{0}$ is a function of $\tau_{S_1}$, i.e.
\begin{equation}
\label{rewriteCons21}
\tau_{0} \geq V(\tau_{S_1}) \triangleq \frac{\tau_{S_1}}{\gamma_{S_1}}(2^{\frac{D_{S_1}}{W\tau_{S_1}}}-1),
\end{equation}
where $ \gamma_{S_1} = \frac{g_{S_1} \zeta_1 P_A h_{S_1}}{WN_0}$. $V(\tau_{S_1})$ is strictly decreasing with respect to $\tau_{S_1}$. 
\end{lemma}
\begin{proof}
We show that the first order derivative of $V(\tau_{S_1})$ is negative to prove the strictly decreasing behavior of $V(\tau_{S_1})$ with respect to $\tau_{S_1}$. 

\begin{equation}
\label{derivativeOfV}
\frac{d V(\tau_{S_1})}{d \tau_{S_1}} = \frac{2^{\frac{D_{S_1}}{W\tau_{S_1}}}-2^{\frac{D_{S_1}}{W\tau_{S_1}}}\frac{D_{S_1}}{W\tau_{S_1}}\ln2-1}{\gamma_{S_1}}.
\end{equation}

The denominator of Eq.~(\ref{derivativeOfV}) is always positive. The numerator of Eq.~(\ref{derivativeOfV}) is always negative since it is a decreasing function of $\frac{D_{S_1}}{W\tau_{S_1}}$, and  it takes its maximum value 0 as $\frac{D_{S_1}}{W\tau_{S_1}}$ goes to 0. Hence, $\frac{d V(\tau_{S_1})}{d \tau_{S_1}} <0$.  





\end{proof}

\begin{lemma} The solution to Problem~(\ref{eq:problm1APNS}) without constraint (\ref{eq:pmax1APNS}) is given by
\begin{align} 
&\dot{\tau}_{S_1} =  \frac{D\ln(2)}{W\alpha_{S_1}}, \label{eq:ITwoR} \\
&\dot{\tau}_0= \frac{D\ln(2)}{W\alpha_{S_1} \gamma_{S_1}} \left(2^{\alpha_{S_1}/\ln(2)}-1\right), \label{eq:EHwoR}
\end{align}
where 
$ \gamma_{S_1} = \frac{g_{S_1} \zeta_1 P_A h_{S_1}}{WN_0}$, $ 
\alpha_{S_1} = \mathbb{L}_0 \left(\frac{\gamma_{S_1}-1}{e}\right) +1$, and $\mathbb{L}_0(.) $ is the Lambert W-function in 0 branch.
\label{lemma:woRelSol}
\end{lemma}
\begin{proof}
 
Let us think of the objective function in the slope-intercept form, e.g., $\tau_{0} = -\tau_{S_1} + C_1$, where $C_1$ is the total transmission time to be minimized. We illustrate the graphical relationship between the curve $V(\tau_{S_1}) = \frac{\tau_{S_1}}{\gamma_{S_1}}(2^{\frac{D_{S_1}}{W\tau_{S_1}}}-1)$ and the line $\tau_{0} = -\tau_{S_1}+C_1$ in Fig. (\ref{fig:cons21}). The area above the curve $V(\tau_{S_1})$ represents the feasible region. The parallel lines represent the different values of $C_1$. There are three possible cases: The line intersects with the curve at two points, at one unique point, i.e., tangent line, and does not intersect at all.

 The intersection point with the tangent line is the optimal solution of the problem. $\tau_{S_1}$ and $\tau_{0}$ values at this unique tangent point are given by Eqs.~(\ref{eq:ITwoR}) and (\ref{eq:EHwoR}), respectively. For the details of this proof, please refer to \cite{Salik19}.
%
\end{proof}

\begin{figure}[!ht]%
\centering
\subfloat[]{\label{fig:cons21}
\includegraphics[width=.5\columnwidth]{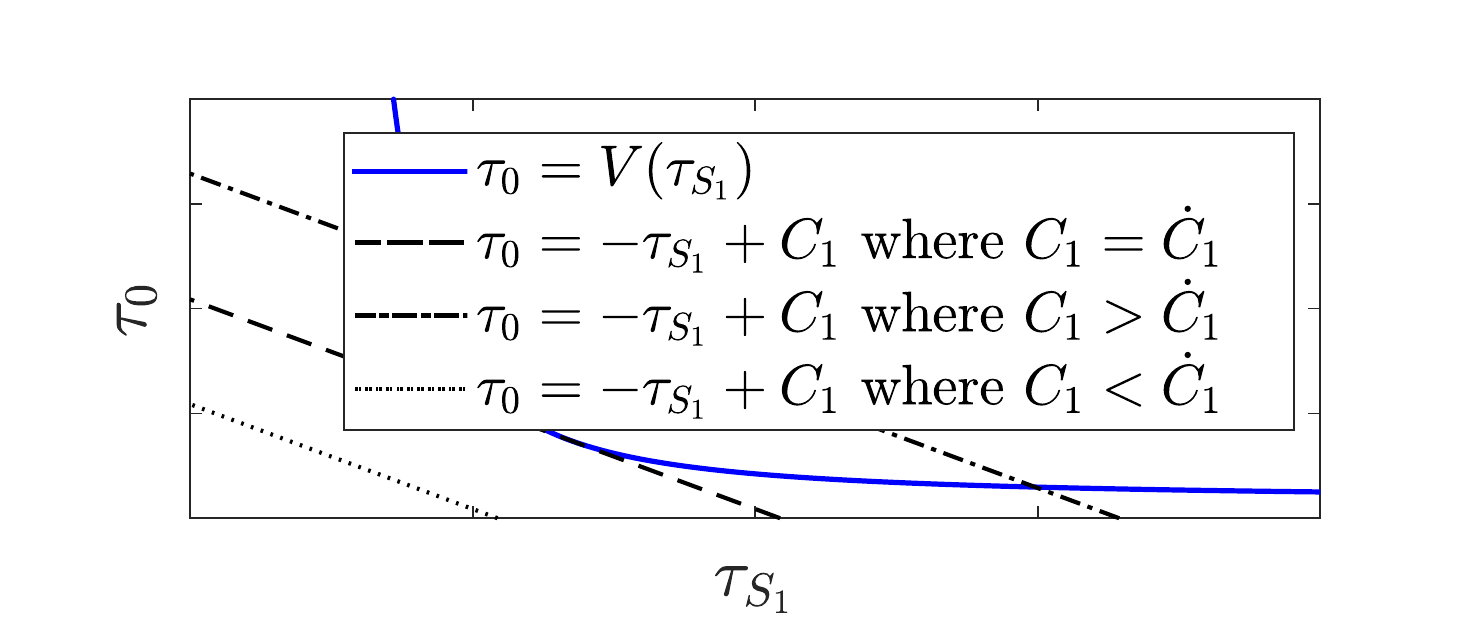}} 
\subfloat[]{ \label{fig:feasreg_pmax}
\includegraphics[width=.5\columnwidth]{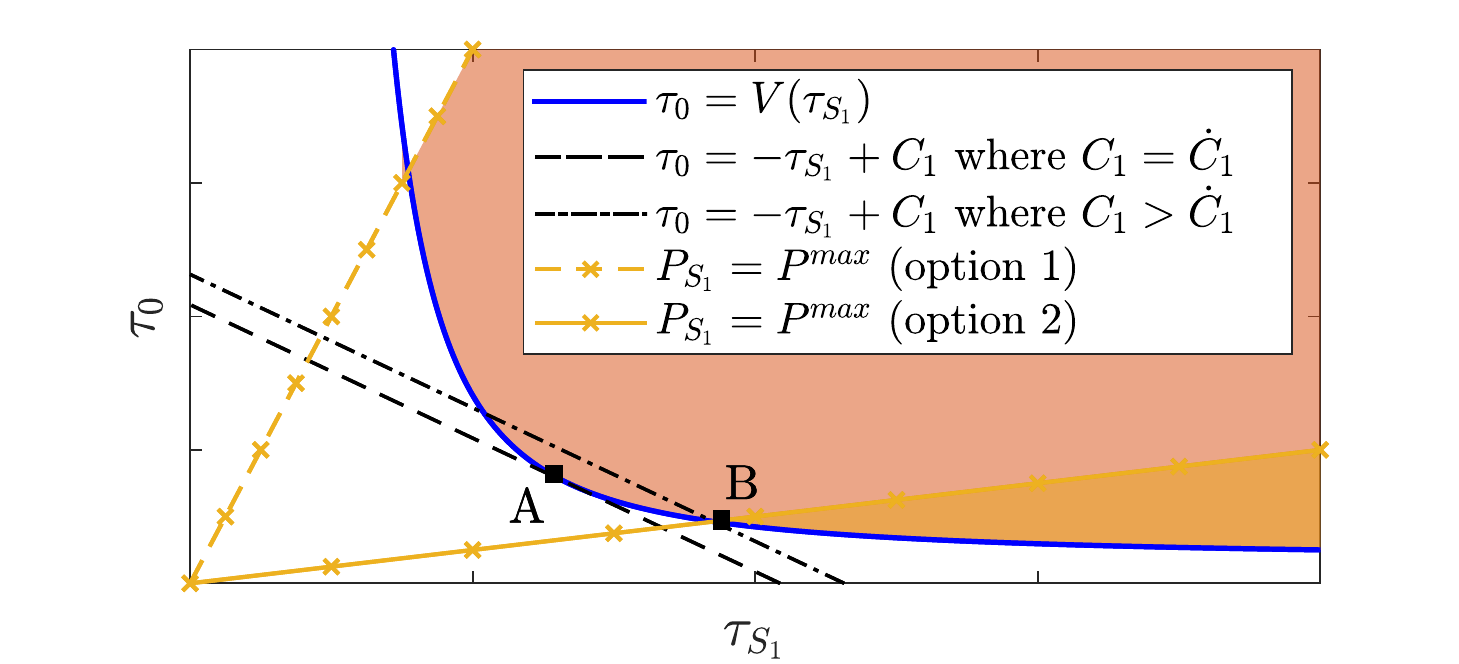}} \\
\caption{(a) Graphical representation of the constraints (\ref{eq:energy1APNS}) and (\ref{eq:data1APNS}) for single source WPCN. (b) Graphical representation of the constraints (\ref{eq:energy1APNS}), (\ref{eq:data1APNS}), and (\ref{eq:pmax1APNS}) for single source WPCN.}
\label{fig:proof_regions}
\end{figure}

\begin{definition} $\ddot{\tau}_{S_1}$ and $\ddot{\tau}_0$ are defined as the values of $\tau_{S_1}$ and $\tau_0$, respectively, in the solution of the system, when constraints (\ref{eq:energy1APNS})-(\ref{eq:pmax1APNS}) hold with equality, and expressed by
\begin{align}
\ddot{\tau}_{S_1} =  \frac{D_{S_1}}{W\log_2\left( 1+\frac{P^{max}g_{S_1}}{WN_0}\right)} \label{eq:ITwoR_pmax} \\
\ddot{\tau}_0 = \frac{P^{max}\ddot{\tau}_{S_1}}{\zeta_1 P_A h_{S_1}}\label{eq:EHwoR_pmax}
\end{align}
\label{def:pmax}
\end{definition}

\begin{theorem}  If the solution $(\dot{\tau}_{S_1}$, $\dot{\tau}_0)$ given in  Eqs.~(\ref{eq:ITwoR}) and (\ref{eq:EHwoR}) satisfies the constraint (\ref{eq:pmax1APNS}) with $P_{S_1} \leq P^{max}$, i.e., $(\dot{\tau}_{S_1}$, $\dot{\tau}_0)$ is feasible, it is optimal to Problem~(\ref{eq:problm1APNS}). Otherwise, the optimal solutions are $\ddot{\tau}_{S_1}$ and $\ddot{\tau}_0$ given in Eqs.~(\ref{eq:ITwoR_pmax}) and (\ref{eq:EHwoR_pmax}), respectively. 
\label{thm:singleSource}\end{theorem}
\begin{proof} We prove the theorem by again using the graphical approach. Fig.~(\ref{fig:feasreg_pmax}) extends the graphical representation in Fig.(~\ref{fig:cons21}) by additionally including the maximum power constraint (\ref{eq:pmax1APNS}). Two possible realizations of (\ref{eq:pmax1APNS}) are depicted by orange lines. The feasible region is the intersection of the area above the blue curve and the area below the orange line. 

Let us consider first that the constraint (\ref{eq:pmax1APNS}) lies as in option 1, where the feasible region is the union of red and orange shaded areas. 
By Lemma~\ref{lemma:woRelSol}, $\dot{\tau}_{S_1}$ and $\dot{\tau}_0$, indicated by point A, minimize the total EH and IT time. Since they also satisfy  $P_{S_1} \leq P^{max}$, as seen in Fig.~\ref{fig:feasreg_pmax}, they are optimal.

Now, let us consider that the constraint (\ref{eq:pmax1APNS}) lies as in option 2, where the feasible region is only the orange shaded area and Point A is not feasible. The optimal solution lies on the boundary of the feasible region \cite{OpResBook} as the problem is convex and the objective is linear. Thus, the intersection of the blue curve and the orange line, point B, minimizes the total schedule length. At point B, the source adjusts its transmit power as $P_{S_1}= P^{max}$, which gives the solution $\ddot{\tau}_{S_1}$ and $\ddot{\tau}_0$ provided in Definition~\ref{def:pmax}.
\end{proof}

\subsection{Multiple Source}

This subsection analyzes a multiple source WPCN, i.e., $N>1$. To solve Problem~(\ref{eq:problm1APNS}) optimally for $N>1$, we transform it into a bi-level optimization problem, in which for a fixed EH duration, the individual IT durations of the sources are solved first, and then, the optimal EH time that minimizes the total schedule length is searched. For the former, we derive the optimal solution for the given EH length; while for the latter, we develop a bi-section search benefiting from the convexity of the total IT times with respect to the EH time. This search obtains a near-optimal solution; and its complexity depends on the gap between the optimal and near-optimal solution.

In the following, we first give the mathematical foundations of this transformation into bi-level optimization in Section~\ref{section:optSol}, and then, present the resulting algorithm in Section~\ref{section:optAlg}. Further, a sub-optimal algorithm based on the solution of the single source case is proposed in Section~\ref{section:suboptAlg} to obtain solutions in a lower time-complexity.

\subsubsection{Bi-level Transformation}
\label{section:optSol}
This subsection provides the mathematical background on the transformation of Problem~(\ref{eq:problm1APNS}) into  bi-level optimization.  The  idea  is to decompose the original  problem into smaller sub-problems, which are then coordinated by a master problem.  For a fixed $\tau_0$, Problem~(\ref{eq:problm1APNS})  is separable into  $N$ \textit{easy} and convex sub-problems regarding each $S_i$ with the aim of minimizing corresponding $\tau_{S_i}$. $\tau_0$ is a variable of master problem and couples the sub-problems. The objective of the master problem is to minimize the sum of $\tau_0$ and the objectives of the sub-problems. The optimal solution to the master problem is optimal for the original Problem~(\ref{eq:problm1APNS}) \cite{Chiang2007}.

 We first formulate the sub-problems and derive the solutions with Lemma~\ref{lemma:solveSubProb}. Theorem~\ref{MultiRecOptSolConv} proves that the objective of the master problem is a convex function of $\tau_0$, which allows us to apply bi-section search on $\tau_0$ for the optimal solution. For this search, the lower and upper bounds on $\tau_0$ are given by Lemmas~\ref{lemma:lb} and \ref{lemma:ub}. 


\textbf{Sub-problems:} All of the sub-problems are in the same reduced form of Problem~(\ref{eq:problm1APNS}) for $N = 1$ and fixed $\tau_0 = \overline{\tau_0} $, which is given as 

\begin{subequations} \label{eq:subproblm1APNS}
\begin{align}
\overline{\tau_{S_1}}(\overline{\tau_0})= \min \qquad &    \tau_{S_1} \label{eq:subprob_obj} \\
\text{s.t.}\qquad & P_{S_1}\tau_{S_1} \leq \zeta_1 P_{A} h_{S_1} \overline{\tau_0} \label{eq:subenergy1APNS}\\
									& \tau_{S_1} W \log_2 \left( 1 + \frac{P_{S_1}g_{S_1}}{WN_0} \right) \geq D_{S_1} \label{eq:subdata1APNS}\\
									& P_{S_1} \leq P^{max} \label{eq:subpmax1APNS}  \\
									& \tau_{S_1}, P_{S_1}   \geq 0 \label{eq:subnonneg1APNS}
\end{align}
\end{subequations}
where $\overline{\tau_{S_1}}(\overline{\tau_0})$ represents the optimal solution. Problem~(\ref{eq:subproblm1APNS}) is non-convex due to the non-convexity of Eqs.~(\ref{eq:subenergy1APNS}) and (\ref{eq:subdata1APNS}). 


\begin{lemma} \label{lemma:solveSubProb} The optimal solution to Problem~(\ref{eq:subproblm1APNS}) is obtained as follows. For $\overline{\tau_0} > \ddot{\tau}_0$, where $\ddot{\tau}_0$ is obtained by Eq.~(\ref{eq:EHwoR_pmax}), the optimal IT duration $\overline{\tau_{S_1}} = \ddot{\tau}_{S_1}$, given by Eq. (\ref{eq:ITwoR_pmax}). Otherwise, $\overline{\tau_{S_1}}$ is derived as a solution of the nonlinear equation
 \begin{equation}
\overline{\tau_0} = V(\overline{\tau_{S_1}}
) = \frac{\overline{\tau_{S_1}}}{\gamma_{S_1}}\left(2^{\frac{D}{W \overline{\tau_{S_1}}}}-1\right)
\label{eq:f_func2}
\end{equation}
\label{lemma:ITforgivenEH}
\end{lemma}

\begin{proof} To prove the lemma, we first derive an equivalent convex formulation of Problem~(\ref{eq:subproblm1APNS}), then, solve the resulting problem via a graphical approach.

For the convex reformulation, let us define $\eta$ as the ratio of the consumed energy in UL to the harvested energy in DL. By using $\eta$, we can rewrite Eq.~(\ref{eq:subenergy1APNS}) as $P_{S_1}\tau_{S_1} = \eta\zeta_1 P_{A} h_{S_1} \overline{\tau_0}$ and combine it with Eq.~(\ref{eq:subdata1APNS}) as
\begin{equation}
\tau_{S_1} W \log_2 \left( 1 + \frac{\eta\zeta_1 P_{A} h_{S_1} \overline{\tau_0}g_{S_1}}{\tau_{S_1}WN_0} \right) \geq D_{S_1}
\label{eq:combEnergyAndData_multi}
\end{equation}

Hence, equivalently, we solve the problem with objective Eq.~(\ref{eq:subprob_obj}) and constraints Eqs.~(\ref{eq:combEnergyAndData_multi}), (\ref{eq:subpmax1APNS}) and (\ref{eq:subnonneg1APNS}). Note that Eq.~(\ref{eq:combEnergyAndData_multi}) represents a convex set for fixed $\overline{\tau_0}$. Moreover, the objective function, Eq.~(\ref{eq:subprob_obj}), and Eqs.~(\ref{eq:subpmax1APNS})-(\ref{eq:subnonneg1APNS}) are linear. Therefore, Problem~(\ref{eq:subproblm1APNS}) turns into a non-linear convex problem.

\begin{figure}[!ht]%
\centering
\includegraphics[width=.5\columnwidth]{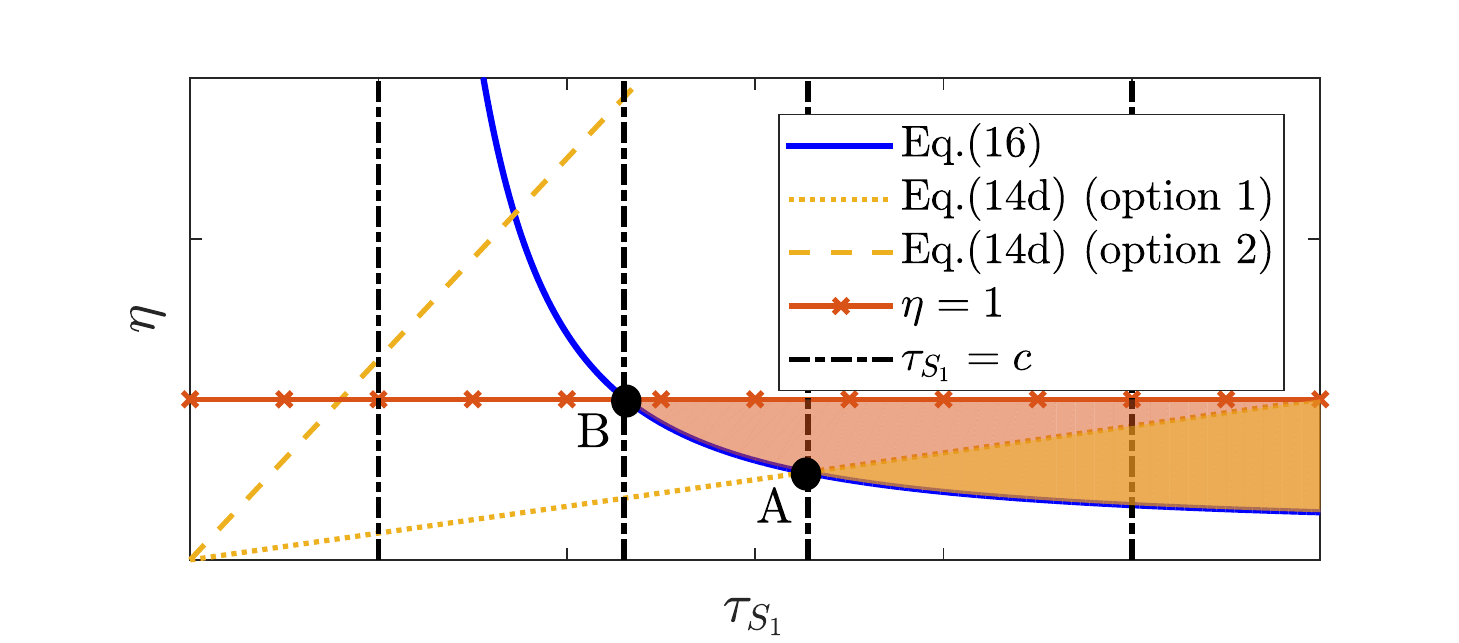}%
\caption{\small Graphical representation of the feasible region for Problem~(\ref{eq:subproblm1APNS}).}
\label{fig:feasibleregion}
\end{figure}

Fig.~\ref{fig:feasibleregion} indicates the feasible region for the problem with the possible values of the objective function $\tau_{S_1}$ (black vertical dotted lines). The blue solid curve represents the left-hand side of the demand constraint Eq.~(\ref{eq:combEnergyAndData_multi}) for fixed $D_{S_1}$ and varying $\eta$ and $\tau_{S_1}$. To satisfy the minimum demand amount $D_{S_1}$, the solution must lie above the curve. On the other hand, due to the energy causality constraint $\eta \leq 1$, a feasible solution must be under the red horizontal dotted line. Depending on maximum UL power constraint (\ref{eq:subpmax1APNS}) (two possibilities are indicated with orange lines) feasible region is either orange area, or both orange and red areas.  

In Fig.~\ref{fig:feasibleregion}, we observe that the feasible region is intersected by  $\tau_{S_1} = c$, where $c$ is a constant,  with minimum $c$ at either Point A or B depending on Eq.~(\ref{eq:subpmax1APNS}).
The optimal solution is at Point A, the intersection of 
Eqs.~(\ref{eq:subpmax1APNS}) and (\ref{eq:combEnergyAndData_multi}), if the feasible region is restricted by Eq.~(\ref{eq:subpmax1APNS}) (option 1). By Definition 1, at Point A, where the source transmits at $P_{max}$,  the transmission necessitates $\ddot{\tau}_{S_1}$ amount of IT and $ \ddot{\tau}_0$ amount of EH duration.  Even though $\tau_0 > \ddot{\tau}_0$, the source transmits at $P_{max}$ by exhausting only a part of the harvested energy to satisfy the constraint  Eq.~(\ref{eq:subpmax1APNS}).

When $\tau_0$ is set to a value that the source cannot harvest enough energy to transmit at $P_{max}$, i.e., $\tau_0 < \ddot{\tau}_0$, Eq.~(\ref{eq:subpmax1APNS}) does not affect the feasible region (option 2). The region is bounded by $\eta =1$, where the source uses all the harvested energy for $D_{S_1}$ amount of data transmission. The optimal solution can be derived by solving Eq.~(\ref{eq:combEnergyAndData_multi}) for $\eta =1$ at equality, which results in the solution of Eq. (\ref{eq:f_func2}).

\end{proof}


\textbf{Master Problem:}
After finding an expression for the individual IT times of the sources for fixed EH length, we now prove that the objective of the master problem, i.e., the total EH and derived IT durations, is convex over EH time, which allows us to search the optimal EH time with bisection method.

\begin{theorem}\label{MultiRecOptSolConv} Let $g({\overline{\tau_{0}}})$ be the objective function of the master problem, i.e., $g({\overline{\tau_{0}}})= \overline{\tau_0} + \sum_i \overline{\tau_{S_i}}(\overline{\tau_{0}})$. $g(\overline{\tau_0})$ is convex over $\overline{\tau_0}$.
\end{theorem}

\begin{proof}
We show that $\overline{\tau_{S_1}}(\overline{\tau_{0}})$ is convex. Then,  $g(\overline{\tau_0})$ is also convex since $\overline{\tau_{S_i}}(\overline{\tau_{0}})$'s are all in the same form and the convexity is closed under sum. To prove the convexity of the $\overline{\tau_{S_1}}(\overline{\tau_{0}})$, we showed in \cite{Salik19} that (a) $\overline{\tau_{S_1}}(\overline{\tau_{0}})$ is proper, (b) the domain of the function denoted by $dom \overline{\tau_{S_1}}$, i.e., the set of all possible $\overline{\tau_{0}}$, is convex and (c) $\overline{\tau_{S_1}}(\overline{\tau_{0}})$ is sub-differentiable at every $\overline{\tau_{0}} \in dom \overline{\tau_{S_1}}$. Then, by Theorem 2.4.1(iii) in \cite{convexVectorAnalysis}, $\overline{\tau_{S_1}}(\overline{\tau_{0}})$ is convex.
\end{proof}


\begin{lemma}\label{lemma:lb}
Let $\{\hat{\tau}_0^i, \hat{\tau}_{S_i}\}$  denote the solution pair of the single source problem corresponding to $S_i$ by Theorem~\ref{thm:singleSource}, for $i = 1,\ldots,N$. $\max_{i = 1,\ldots,N} \hat{\tau}_0^i$ gives a lower bound for the optimum $\tau_0$ of Problem~(\ref{eq:problm1APNS}).
\end{lemma}
\begin{proof}
To prove the statement, we discuss that further decrease of $\tau_0$ below $\max_{i = 1,\ldots,N} \hat{\tau}_0^i$ would increase the total schedule length. Without loss of generality, suppose  $\hat{\tau}_0^1 \geq \hat{\tau}_0^i, \forall i, i\neq 1$.

 Let us set $\tau_0$ to $\hat{\tau}_0^1$. The pair of $\{\hat{\tau}_0^1, \hat{\tau}_{S_1}\}$ minimizes $\tau_0 +\tau_{S_1}$ by Theorem~\ref{thm:singleSource}.  For a fixed $\tau_0 = \hat{\tau}_0^1$, optimal transmission time of source $i$ is denoted by $\hat{\tau}_{S_i}(\hat{\tau}_0^1)$.
 By Lemma~\ref{lemma:f}, $\tau_0$ is a decreasing function of $\tau_{S_i}$. Then, by inverse function theorem \cite{Trench2013}, $\tau_{S_i}(\tau_0)$ is a decreasing function of $\tau_0$. Intuitively, with more harvested energy, sources can complete transmission in shorter time. Hence, $\tau_{S_i}(\hat{\tau}_0^1) \leq \hat{\tau}_{S_i} = \tau_{S_i}(\hat{\tau}_0^i), \forall i$ since $\hat{\tau}_0^1 \geq \hat{\tau}_0^i$. 
The objective is minimizing the total schedule length  $\tau_0 + \sum_{i=1}^{N} \tau_{S_i}$. If $\tau_0 < \hat{\tau}_0^1$ is picked, the term  $\tau_0 +\tau_{S_1}$ and all $\tau_{S_i}$s, so, the total schedule length increases. If $\tau_0 > \hat{\tau}_0^1$ is picked, despite the increase in the term  $\tau_0 +\tau_{S_1}$, $\tau_{S_i}$s $\forall i, i \neq 1$ decrease, which may result in a decrease of the total schedule length.  Hence, $\max_{i = 1,\ldots,N} \hat{\tau}_0^i$ gives a lower bound on optimal $\tau_0$.
\end{proof}

\begin{lemma}\label{lemma:ub}
Let $\ddot{\tau}_0^i$ be the solution of single source network when constraints (\ref{eq:energy1APNS})-(\ref{eq:pmax1APNS}) hold with equality with the only source $S_i$, for $i = 1,\ldots,N$.  Define $\ddot{\tau}_0^{max} = \max_{i \in \{1,\ldots,N\}}\{\ddot{\tau}_0^i\} $. Then, $\ddot{\tau}_0^{max}$ is the upper bound for the optimum $\tau_0$ of Problem~(\ref{eq:problm1APNS}). 
\end{lemma}
\begin{proof}
When $\overline{\tau_{0}} > \ddot{\tau}_0^i$, $\tau_{S_i} $ is constant for all $i$, i.e., $\tau_{S_i} =\ddot{\tau}_{S_i}$, by Lemma~\ref{lemma:solveSubProb}. Then, further increase of $\overline{\tau_{0}}$  greater than  $\ddot{\tau}_0^{max}$ does not improve  
$\tau_{S_i}$ for any $i$. Since the objective function is minimizing  total EH and IT times  $\tau_0 + \sum_i \tau_{S_i}$,  $\ddot{\tau}_0^{max}$ gives the upper bound on the optimal $\tau_0$. 
\end{proof}

\subsubsection{Optimal Algorithm} \label{section:optAlg} 

\begin{algorithm}[H]
	\caption{\small \textbf{Power Constrained Multiple User Time Minimization Algorithm (POWMU)}}\label{alg:mrttma}
	\begin{algorithmic}[1]
		\small\FORALL{ $i =1,\ldots,N$}
		\small\STATE Compute $\dot{\tau}_{S_i}$ and $\dot{\tau}_0^i$  by Eqs.~(\ref{eq:ITwoR}) and (\ref{eq:EHwoR}) for source $S_i$.
	\small\STATE Compute $\ddot{\tau}_0^i$ by Eq.~(\ref{eq:EHwoR_pmax}) for source $S_i$. 
        \small\STATE \textbf{if } {$\zeta_i P_A h_{S_i} \dot{\tau}_0^i / \dot{\tau}_{S_i} \leq P_{max}$}, \textbf{then } $\hat{\tau}_0^i \leftarrow \dot{\tau}_0^i $;
        \small\STATE \textbf{else }  $\hat{\tau}_0^i \leftarrow \ddot{\tau}_0^i $.
		\small\ENDFOR
		\small\STATE $lb \leftarrow \max_{i \in \{1,\ldots,N\}} \hat{\tau}_0^i$
		\small\STATE $ub \leftarrow \max_{i \in \{1,\ldots,N\}} \ddot{\tau}_0^i$
		\small\WHILE {$ub - lb > 2\epsilon$}		
		\small\STATE $\tau_{0}^{'} = \frac{lb + ub}{2}$		
		\small\STATE \textbf{if }  $\frac{dg(\overline{\tau_0})}{d(\overline{\tau_0})}\rvert_{\overline{\tau_0} = \tau_{0}^{'}} \geq 0$,  \textbf{then } $ub \leftarrow \tau_{0}^{'}$.
		\small\STATE \textbf{if }  $\frac{dg(\overline{\tau_0})}{d(\overline{\tau_0})}\rvert_{\overline{\tau_0} = \tau_{0}^{'}} \leq 0$, \textbf{then }  $lb \leftarrow \tau_{0}^{'}$.
		\small\ENDWHILE
		\small\STATE Evaluate and return $g({\overline{\tau_{0}}})$ at $\tau_{0}^{'}$
	\end{algorithmic}
\end{algorithm}
Based on the bi-level transformation, POWMU, given in Algorithm~\ref{alg:mrttma}, obtains the optimal solution as follows. The for loop (Lines 1-6) computes $\hat{\tau}_0^i$ and $\ddot{\tau}_0^i$ for all $i=1,\ldots,N$ by following Theorem~\ref{thm:singleSource} and Definition~\ref{def:pmax}, respectively. Then, the lower ($lb$) and upper ($ub$) bounds on $\tau_0$ are set as maximum of  $\hat{\tau}_0^i$'s and  $\ddot{\tau}_0^i$'s  based on Lemmas~\ref{lemma:lb}  and \ref{lemma:ub}, respectively  (Lines 7-8). The algorithm continues with the bi-section search on $\tau_0$. In each iteration, the current solution, $\tau_0 ^{'}$, is set to the mean of the $lb$ and $ub$ (Line 10). If the first derivative of $g(\overline{\tau_0})$ calculated at $\tau_0 ^{'}$ is non-negative, meaning that the function increases at that point, then the $ub$ is set to $\tau_0 ^{'}$ (Line 11). If the derivative is non-positive, the function is decreasing at that point, the $lb$ is  set to $\tau_0 ^{'}$ (Line 12). This procedure continues until the difference between $ub$ and $lb$ is lower than $2\epsilon$ (Line 9). The last computed $ \tau_0 ^{'} $ is set as the optimal EH time and the objective of the master problem, $ g(\overline{\tau_0})$, is evaluated at $\tau_0 ^{'}$ and returned (line 14).

%

\begin{remark}\label{BisectionComplexity} Given the initial $lb$ and $ub$, by bisection search method, to obtain an optimal solution in the error bound, $\epsilon$, the number of iterations, $k$, should be at least as large as $log_2(\frac{ub-lb}{\epsilon})$ \cite{Salik19}.\end{remark}

The overall complexity of POWMU is $\mathcal{O}(Nlog_2(\frac{ub-lb}{\epsilon}) )$, with the derivative calculations, $\mathcal{O}(N)$, at each iteration and by Remark \ref{BisectionComplexity}.

\subsubsection{Sub-optimal Algorithm}\label{section:suboptAlg} In this subsection, we present a sub-optimal algorithm to obtain solutions in lower time-complexity. The sub-optimal algorithm MAX-EH, given in Algorithm~\ref{alg:maxEH}, provides a solution to the multiple source problem by exploiting the optimality conditions of single source network. After individual optimum schedule lengths of the sources are obtained, the EH duration of the entire network is set to the maximum of the individual EH durations.  

The algorithm is described as follows. Individual optimal EH length of each source $S_i$, expressed as $\hat{\tau}_0^i$, is calculated  by Theorem~\ref{thm:singleSource} (Lines 1-6). Maximum of these solutions, denoted by $\hat{\tau}_{0}^{max}$, is computed (Line 7). Then,  IT durations, $\tau_{S_i}$, for fixed  $\tau_0 =\hat{\tau}_{0}^{max}$, are obtained for all $i$ by following Lemma~\ref{lemma:ITforgivenEH} (Lines 8-11). The objective function value is computed as a sum of $\hat{\tau}_{0}^{max}$ and $\tau_{S_i}$s and returned (Line 12).

The time complexity of the algorithm is $\mathcal{O}(N)$ as obtaining $\hat{\tau}_0^i$s and finding the maximum require $N$ iterations.  

\begin{remark}\label{rem:maxeh}
The solution obtained by MAX-EH provides a tight upper bound on the optimal solution of   Problem~(\ref{eq:problm1APNS}) for moderate and high AWGN levels.  We show the tightness of this upper bound by numerical simulations. The detailed simulation parameters are given in Section~\ref{section:simRes}. The optimality gap of the algorithm, versus varying number of sources ($N$) and AWGN spectral density ($N_0$), is depicted in Fig.~\ref{fig:optgap}.  In general, the gap is lower than $5\%$; whereas it is lower for the small number of sources and 0 for single source network as expected. For $N_0 \geq -90$ dBm, moderate and high noise level, the optimality gap is lower than $0.2\%$. This is because the prolonged IT durations due to the higher noise level become dominant in the total schedule length. 
\begin{figure}[!ht]%
\centering
\includegraphics[width=.4\columnwidth]{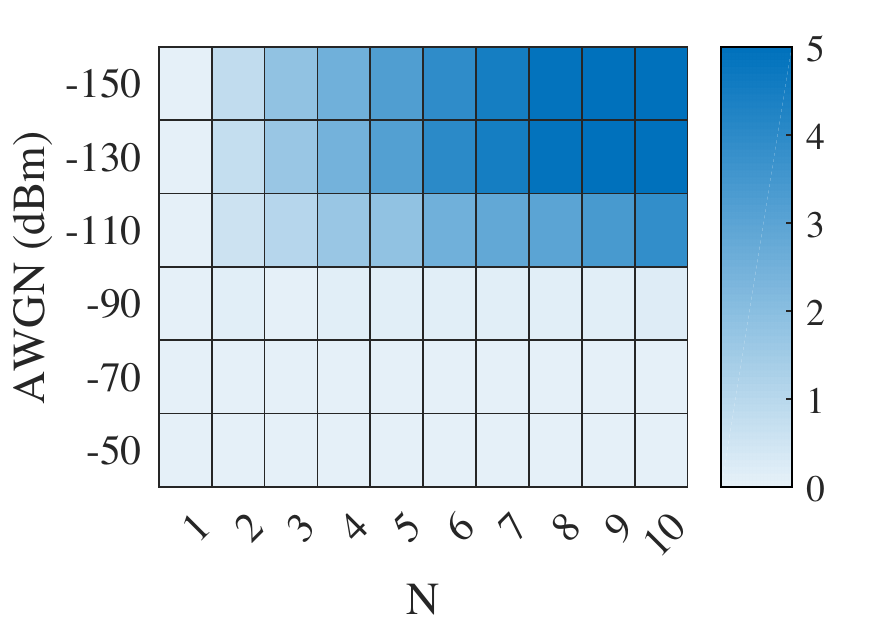}%
\caption{Optimality gap ($\%$) of MAH-EH for varying number of sources ($N$) and AWGN power spectral density ($N_0$)}%
\label{fig:optgap}%
\end{figure}
\end{remark}

\begin{algorithm}[H]
	\caption{\small \textbf{MAX-EH}}\label{alg:maxEH}
	\begin{algorithmic}[1]
        \small\FORALL{ $i =1,\ldots,N$}
		\small\STATE Compute $\dot{\tau}_{S_i}$ and $\dot{\tau}_0^i$  by Eqs.~(\ref{eq:ITwoR}) and (\ref{eq:EHwoR}) for source $S_i$.
    	\small\STATE Compute $\ddot{\tau}_0^i$ by Eq.~(\ref{eq:EHwoR_pmax}) for source $S_i$. 
        \small\STATE \textbf{if } {$\zeta_i P_A h_{S_i} \dot{\tau}_0^i / \dot{\tau}_{S_i} \leq P_{max}$}, \textbf{then } $\hat{\tau}_0^i \leftarrow \dot{\tau}_0^i $;
        \small\STATE \textbf{else }  $\hat{\tau}_0^i \leftarrow \ddot{\tau}_0^i $.
		\small\ENDFOR
		\small\STATE $\hat{\tau}_{0}^{max} = \max_{i = 1,..., N}\{\hat{\tau}_0^i\}$ 
		\small\FORALL{$i =1,\ldots,N$}
		 \small\STATE \textbf{if } $\max_{i = 1,..., N}\{\hat{\tau}_0^i\} > \ddot{\tau}_0^i$, \textbf{then } $\tau_{S_i} \leftarrow \ddot{\tau}_{S_i}$ by Eq.~(\ref{eq:ITwoR_pmax})
		 \small\STATE \textbf{else } Obtain $\tau_{S_i}$ by Eq.~(\ref{eq:f_func2}).
		\small\ENDFOR
		\small\STATE Compute and return: $ \hat{\tau}_{0}^{max} + \sum_{i=1}^{N} {\tau_{S_i}}$.
	\end{algorithmic}
\end{algorithm}
\section{Branch and Bound Based Algorithm} \label{section:relsel}
In this section, we develop a Branch-and-Bound based Algorithm (BBA) for the optimal solution of Problem~(\ref{eq:relSel}). Branch-and-Bound (BB) is a tree based search algorithm and BB based techniques are commonly used for the solution of non-convex MINLPs \cite{PietroBelotti2012}.  They systematically explore the solution space by dividing it into smaller sub-spaces, i.e., branching. The entire solution space is represented by the root of the BB-tree and   each resulting sub-space after partitioning is represented by a BB-tree node. Before creating new branches, the current BB-tree node is checked against lower and upper bounds on the optimal solution and is discarded if a better solution than the incumbent solution cannot be obtained, i.e, pruning.  We adapt BB for our problem by developing problem specific lower and upper bound generation methods and integrating it with POWMU.  

BBA starts with the MINLP formulation given in Problem~(\ref{eq:relSel}). The relaxation of the problem is obtained and solved. If $b_i^j$s  in the solution of the relaxed problem are fractional, BBA continues with branching. A source is selected by $i' = \argmax_{i,j, i\in \{1,\ldots,N\}, j\in \{0,\ldots,K\}}b_i^j$.  Then, branching is performed by assigning one of the relays $R_j, j=0,\ldots,K$ to  the $i^{'th}$ source at each new branch. As a result, $K+1$ new nodes are created, where each new node corresponds to the $j^{th}$, $j=0,\ldots, K$ , relay  selected for the $i^{'th}$ source. This selection is forced by setting $ b_{i'}^j$ to 1 if the $j^{th}$ relay is selected, and 0 otherwise. It is possible to continue branching on other variables, however, this is not necessary as after relays are selected, the problem can be solved by POWMU. 

During the BB-search, the solution of the relaxed problem is used in the branching decision and the objective value of the relaxed problem is considered as a lower bound on the objective of the original problem. We derive the relaxed problem and the lower bound of Problem~(\ref{eq:relSel}) in Section~\ref{section:rp}. Moreover, an incumbent solution is stored. To update the incumbent solution, in each node to be branched, an upper bound is calculated. The upper bound can be any feasible solution in the sub-space. If the upper bound of a node is lower than the incumbent, then it is set as the incumbent. We present the upper bound generation for Problem~(\ref{eq:relSel}) in Section \ref{section:ub}.

A BB-tree node is pruned, if the further search of the corresponding sub-space is not necessary. We prune a node if a feasible solution cannot be found at that node; the lower bound of the node is greater than the incumbent solution, as the node cannot provide a better solution than the current one; or all $b_i^j$s of the solution are integers, as we call POWMU to obtain the minimum schedule length for the determined relay selection, and this is the optimum solution in the corresponding sub-space. The algorithm continues until all BB-nodes are pruned.


\subsection{Lower Bound Generation} \label{section:rp}
A convex relaxation of the sub-problem is solved to optimality to find a lower bound on the objective. The relaxation of Problem~(\ref{eq:relSel}) is obtained as follows. 
\begin{itemize}
\item  The integrality constraint Eq.~(\ref{eq:relSel:integrality}) of Problem~(\ref{eq:relSel}) is relaxed as $0 \leq b_i^j \leq 1$.
\item In Problem~(\ref{eq:relSel}), $P_{S_i}^{R_j} \tau_{S_i}^{R_j}$ and $P_{R_j}^{AP} \tau_{R_j}^{AP}$ are replaced by new variables $A_{S_i}^{R_j}$ and $A_{R_j}^{AP}$, respectively, so that constraints (\ref{eq:relSel:energy1})-(\ref{eq:relSel:demand2}) represent convex sets.

\item After the replacement,  Eqs.~(\ref{eq:relSel:maxpow1}) and (\ref{eq:relSel:maxpow2}) include a product of two variables (e.g., $b_i^j \tau_{S_i}^{R_j}$ and $b_i^j \tau_{R_j}^{AP}$), which causes the non-convexity. Instead of these products, we define their convex/concave envelopes, lower and upper bounding constraints, which is a common relaxation technique proposed in \cite{PietroBelotti2012}.

\end{itemize}
Now, the relaxed problem can be solved by a convex optimization tool to find the lower bound, e.g., CVX \cite{cvx}. 

\subsection{Upper Bound Generation}\label{section:ub}
For each node, any feasible solution to the problem in the corresponding sub-space can be an upper bound. The initial upper bound, also the initial incumbent solution, is obtained by assigning all sources directly to AP, which is the simplest feasible solution. In other BB-tree nodes, we obtain the upper bound as follows. The relaxed problem is already solved for lower bound generation and the resulting $b_i^j$s are known. For each source $S_i$, the relay with the maximum $b_i^j$ is selected, i.e., $R_ j =\argmax_ {j\in\{0,1,\ldots,K\}} b_i^j, \; \; i =1, \ldots, N$. Then, POWMU acquires a feasible solution resulting in the minimum schedule length for the determined relay selection.

 \section{Branch-and-Bound Based Heuristics} \label{section:sBB2}
BBA searches for all possible relay selection options, $K^N$ nodes, in the worst case, which results in an exponential run-time. We now propose the BB-based heuristic algorithms for better time complexity. 

\subsubsection{One Branch Heuristic (OBH)} OBH explores only one branch of the Branch-and-Bound tree. It starts by obtaining the relaxation of Problem~(\ref{eq:relSel}), as described in Section~\ref{section:rp}.  The relaxed problem is solved to obtain initial fractional values of  $b_i^j$s.  Among the fractional $b_i^j$s,  the maximum one is selected and it is forced to be 1 in the next iteration.  OBH repeatedly solves the relaxed problem containing $b_i^j$ values set to 1 in the previous iterations, finds the maximum of the fractional $b_i^j$ in the optimal solution of the resulting problem and forces it to be 1. This continues until the relay selection of all the source nodes. Finally, POWMU provides the solution for the resulting relay selection. 

In OBH, $N-1$ relaxed problems are solved and POWMU is called only once. The complexity of solving the relaxed problem depends on the solver and is ambiguous. It is highly dependent on the number of variables, so rapidly increases with increasing number of sources or relays. Let us denote this complexity with $\mathcal{O}(C)$. The complexity of POWMU is given in Section~\ref{section:optAlg}. Thus, the overall complexity is $\mathcal{O}((N+1)C + Nlog_2(\frac{ub-lb}{\epsilon}))$
\subsubsection{Relaxed Problem Based Heuristic (RPH)}: RPH considers only the relaxed problem and does not perform branching. It obtains and solves the relaxation of Problem~(\ref{eq:relSel}) as in Section~\ref{section:rp}. This solution gives fractional $b_i^j$ values. Then, the relay $R_j$ with  $\max_{j\in \{0,\ldots,K\}} b_i^j$  is assigned to each source $S_i$, for $i =1, \ldots, N$.  After relay selection, the total EH and IT times is obtained by POWMU. Note that this is the same approach that we use to determine upper bounds of each node in BBA. 

RPH solves only one relaxed problem and calls  POWMU once. The overall complexity is $\mathcal{O}(C + Nlog_2(\frac{ub-lb}{\epsilon}))$

\section{Relay Criterion Based Heuristic Algorithm} \label{seciton:relCritHeur}

The ultimate aim of this section is to propose a heuristic algorithm of lower complexity for the solution of Problem~(\ref{eq:relSel}). Although the BB-based heuristic approaches have polynomial time-complexity, their runtime rapidly increases with the number of variables as they must solve the relaxed problem. For large size networks, the need for another heuristic approach with lower time-complexity arises. 

We first analyze the simpler problem instances to derive a CSI-based relay selection condition of three node WPCCN in Section~\ref{section:analyze}. In Section~\ref{section:RSTTMA}, we extend the condition for a single source multiple relay network for relay selection, and then, present the heuristic algorithm, which initializes relay assignment based on this condition and iteratively updates the relay assignment as long as there is an improvement in the schedule length. 


\subsection{Optimality Conditions}\label{section:analyze}
In three node WPCCN, the condition on whether a relay assistance is beneficial or not is established by comparing the total schedule lengths of cases where IT is performed with and without a relay. Obviously, the network benefits from the relay if the total schedule length with relay assistance is lower than that of direct source-to-AP communication.

For the following Lemma~\ref{lemma:dec} and Claim~\ref{lemma:relCond}, we assume that $P^{max}$ is sufficiently high so that the power constraints Eqs.~(\ref{eq:relSel:maxpow1}) and (\ref{eq:relSel:maxpow2}) are not effective to ease the mathematical comparisons. Later, we will discuss the effect of $P^{max}$ over an example scenario.

\begin{lemma}  \label{lemma:dec} Let $\{\hat{\tau}_0^i, \hat{\tau}_{S_i}^{AP}\}$  be the solution pair of the single source problem corresponding to $S_i$ for direct transmission to AP by Theorem~\ref{thm:singleSource} for $i = 1,\ldots,N$. $\hat{\tau}_0^i$ and $\hat{\tau}_{S_i}^{AP}$ are decreasing functions of $g_{S_i}^{AP}h_{AP}^{S_i}$ \label{lemma:EHandGamma}. 
\end{lemma}

\begin{proof}
We prove the lemma by showing that the first derivatives of $\hat{\tau}_0^i$ and $\hat{\tau}_{S_i}^{AP}$ with respect to $g_{S_i}^{AP}h_{AP}^{S_i}$ are negative. Due to the assumption on $P^{max}$, Eqs.~(\ref{eq:ITwoR}) and (\ref{eq:EHwoR}) determine $\{\hat{\tau}_0^i, \hat{\tau}_{S_i}^{AP}\}$. Note that $g_{S_i}^{AP}h_{AP}^{S_i}$ is proportional to $\gamma_{S_i}$, as  $\gamma_{S_i} = \frac{g_{S_i}^{AP} \zeta_i P_A h_{AP}^{S_i}}{WN_0}$. The first derivative of Eq.~(\ref{eq:ITwoR}) is 
\begin{equation}
\frac{d \hat{\tau}_{S_i}^{AP}}{d \gamma_{S_i}} = -\frac{\mathbb{L}_0 \left(\frac{\gamma_{S_i}-1}{e}\right)}{(-1+\gamma_{S_i})\left(1+\mathbb{L}_0 \left(\frac{\gamma_{S_i}-1}{e}\right)\right)^3}
\label{eq:derivIT}
\end{equation}
$\mathbb{L}_0(.) $ is the Lambert W-function in 0 branch. If $\gamma_{S_i} > 1$, , then $(-1+\gamma_{S_i}) > 0$ and $\mathbb{L}_0 \left(\frac{\gamma_{S_i}-1}{e}\right) >0$. Thus, ${d \bar{\tau}_{S_i}^{AP}}/{d \gamma_{S_i}} <0$.
 If $ 0 < \gamma_{S_i} < 1$, then $(-1+\gamma_{S_i}) < 0$ and $\mathbb{L}_0 \left(\frac{\gamma_{S_i}-1}{e}\right) <0$. Also, $-1 < \mathbb{L}_0 \left(\frac{\gamma_{S_i}-1}{e}\right) <0$, since $-1/e < (-1+\gamma_{S_i})/ e <0$. Thus, ${d \hat{\tau}_{S_i}^{AP}}/{d \gamma_{S_i}} <0$. As a result, $\hat{\tau}_{S_i}^{AP}$ is a decreasing function of $g_{S_i}^{AP}h_{AP}^{S_i}$.

To ease the calculations, the first derivative of Eq.~(\ref{rewriteCons21}) is derived instead of Eq.~(\ref{eq:EHwoR}).
\begin{equation}
\frac{d \hat{\tau}_0}{d \gamma_{S_i}} = \frac{\hat{\tau}_{S_i}}{\gamma^2_{S_i}} \left(1- 2^{\frac{D}{W \hat{\tau}_{S_i}}} \right) <0
\end{equation}
since $\frac{D}{W\tau_{S_i}} >0$ so $2^{\frac{D}{W\tau_{S_i}}} >1$. As a result, $\hat{\tau}_{0}^i$ is a decreasing function of $g_{S_i}^{AP}h_{AP}^{S_i}$.
\end{proof}

\begin{claim} \label{lemma:relCond} Under the moderate and high AWGN, the necessary condition for relay selection $R_1$ at source $S_1$ is 
\begin{equation}
	\min \left(g_{S_1}^{R_1}h_{AP}^{S_1}, g_{R_1}^{AP}h_{AP}^{R_1} \right) > g_{S_1}^{AP}h_{AP}^{S_1}
	\label{eq:relSelCondGeneral} 
\end{equation}
\end{claim}
\begin{proof} 
We prove the lemma by comparing the total schedule length of the cases with and without relay assistance. To compute the schedule lengths, we benefit from the heuristic algorithm MAX-EH instead of optimal algorithm  POWMU based on Remark~\ref{rem:maxeh}. 

By following MAX-EH solution strategy, the schedule lengths are calculated as follows. If there is no relay assistance, the solution  for the direct transmission, $\{\hat{\tau}_0, \hat{\tau}_{S_1}^{AP}\}$ , is obtained by Eqs.~(\ref{eq:ITwoR}) and (\ref{eq:EHwoR}). We denote the solution of three node WPCCN by $\{\tau'_0 , \tau_{S_1}^{'R_1}, \tau_{R_1}^{'AP}\}$. Let $\bar{\tau}_0^{S_1}$ [ $ \bar{\tau}_0^{R_1} $] be the optimal EH length when $S_1$ [$R_1$] is the single source targeting the $R_1$ [the AP] of the network.  $\tau'_0$ is set to the maximum of $\{\bar{\tau}_0^{S_1},  \bar{\tau}_0^{R_1} \}$; and $\tau_{S_1}^{'R_1}$ and $\tau_{R_1}^{'AP}$ are calculated for given $\tau'_0$ by Eq~(\ref{eq:f_func2}).
 
 Now, we show the condition where $\tau'_0 + \tau_{S_1}^{'R_1}+ \tau_{R_1}^{'AP}  <\hat{\tau}_0 + \hat{\tau}_{S_1}^{AP}$. Let us first consider the case where the source needs more energy than the relay, i.e., $\bar{\tau}_0^{S_1} > \bar{\tau}_0^{R_1}$.  By Lemma~\ref{lemma:EHandGamma}, this refers to $g_{S_1}^{R_1}h_{AP}^{S_1} < g_{R_1}^{AP}h_{AP}^{R_1}$.  $\tau'_0$ is set to $\bar{\tau}_0^{S_1}$  so it is a decreasing function of $g_{S_1}^{R_1}h_{AP}^{S_1}$. In direct transmission, $\hat{\tau}_0$ is decreasing function of $ g_{S_1}^{AP}h_{AP}^{S_1}$. Again based on Lemma~\ref{lemma:EHandGamma},   if $g_{S_1}^{AP}h_{AP}^{S_1} < g_{S_1}^{R_1}h_{AP}^{S_1}$, $\tau'_0 < \tau_0$ and $\tau_{S_1}^{'R_1} < \tau_{S_1}^{AP}$ hold. Otherwise, relay cannot bring improvement since already $\tau_0 + \tau_{S_1}^{AP} < \tau'_0 + \tau_{S_1}^{'R_1}$. As a result, in the case of $g_{S_1}^{R_1}h_{AP}^{S_1} < g_{R_1}^{AP}h_{AP}^{R_1}$, the necessary condition of relay selection is $g_{S_1}^{AP}h_{AP}^{S_1} < g_{S_1}^{R_1}h_{AP}^{S_1}$. 

Next we analyze the case where the source needs less energy than the relay, i.e., $\bar{\tau}_0^{S_1} < \bar{\tau}_0^{R_1} $. By Lemma~\ref{lemma:EHandGamma}, this refers to  $g_{S_1}^{R_1}h_{AP}^{S_1} > g_{R_1}^{AP}h_{AP}^{R_1}$. Similar to the previous case, inequality $g_{S_1}^{AP}h_{AP}^{S_1} < g_{R_1}^{AP}h_{AP}^{R_1}$ is the necessary condition for relay selection.
By considering both cases, the minimum of $\left(g_{S_1}^{R_1}h_{AP}^{S_1}, g_{R_1}^{AP}h_{AP}^{R_1} \right)$ must be greater than $g_{S_1}^{AP}h_{AP}^{S_1}$ for a relay to bring improvement.
\end{proof}


\begin{figure}[!ht]%
\centering
\subfloat[]{\label{fig:relCopArea_a}
\includegraphics[width=.5\columnwidth]{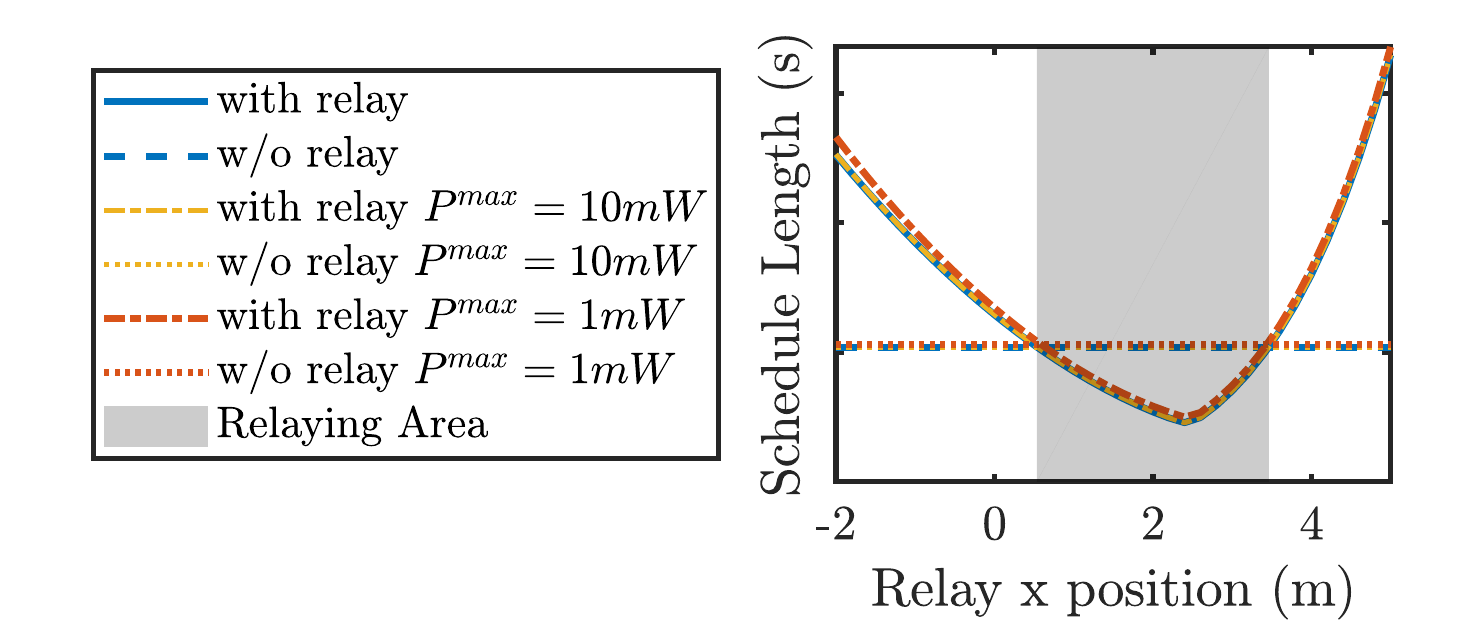}} 
\subfloat[]{ \label{fig:relCopArea_b}
\includegraphics[width=.5\columnwidth]{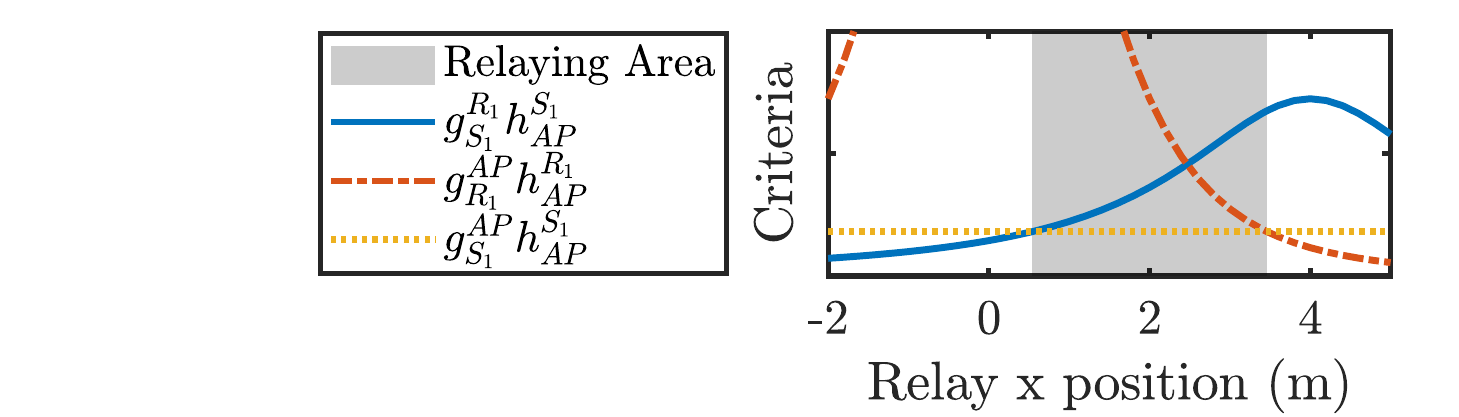}} \\
\caption{Example to validate Eq.~(\ref{eq:relSelCondGeneral}) as a relay selection criterion. (a) Schedule length vs. the relay position in three node WPCCN. (b) Components of the criterion given by Eq.~(\ref{eq:relSelCondGeneral}) vs. the relay position.}
\label{fig:relCopArea}
\end{figure}

Now, we exemplify the relay selection condition given in Eq.~(\ref{eq:relSelCondGeneral}). In this example,  an AP and a source are located at (0,0) and (4,0) in a gridded area, respectively.  The relay is moved  on the horizontal axis from (-2,2) to (5,2).  We assume that the channels gains are only dependent on the distance. The rest of the simulation parameters are given in Section~\ref{section:simRes}. Under moderate noise level, variation of the total transmission time as relay moves is depicted in Fig.~(\ref{fig:relCopArea_a}). The dashed blue line is the total time when information is transferred  without the relay; whereas the solid blue line is the total time when information is transferred via the relay. By comparing blue lines, it is observed that transmission with the relay's help results in lower total time when relay is in between $x=0.53592$ and $x=3.46408$. The shaded area is the relaying region determined by Eq.~(\ref{eq:relSelCondGeneral}). The distance between the intersection point of blue lines and the closest edge to that point of the shaded area is only $2\times 10^{-5}$m. This demonstrates that Eq.~(\ref{eq:relSelCondGeneral}) gives the necessary condition for relay selection at moderate noise level.


In Fig.~(\ref{fig:relCopArea_a}), the orange and red lines correspond to the results for different $P^{max}$ values. As $P_{max}$ decreases, the error in the shaded area increases, however, we observe that in realistic power levels such as 10 mW, the error is negligible, e.g., $\%0.4$.




\subsection{Heuristic Algorithm}\label{section:RSTTMA}

The straightforward approach to determine source-relay pairs for the solution of Problem~(\ref{eq:relSel}) is choosing the individual best relay for each source.  Then, the total schedule length can be calculated by POWMU.  However, this is not necessarily the best selection for the entire network.  If one specific relay is assigned to many sources, then the selected relay needs to forward many messages. This requires more EH time resulting in an increase in the schedule length of the network. Therefore, next, we develop a heuristic algorithm, called RSTMA, which starts with an initial solution obtained by individual relay selection and reassigns a source from a relay to another in each iteration as long as there is improvement in the total schedule length.

The individual relay selection is performed as follows. In Fig.~(\ref{fig:relCopArea_b}), the blue line indicates the criterion $g_{S_1}^{R_1}h_{AP}^{S_1}$; whereas the red line represents the criterion $g_{R_1}^{AP}h_{AP}^{R_1}$. We observe that the higher $\min \left(g_{S_i}^{R_j}h_{AP}^{S_i}, g_{R_j}^{AP}h_{AP}^{R_j} \right)$ value, the lower the total EH and IT time. Accordingly, the individual relay selection criterion can be defined as 
\begin{equation}
\argmax_j \min \left(g_{S_i}^{R_j}h_{AP}^{S_i}, g_{R_j}^{AP}h_{AP}^{R_j} \right)
\label{eq:relaySelectCriteria}
\end{equation}
for source $S_i$. Note that this criterion is same as the opportunistic relaying (OR) criterion proposed for outage performance in \cite{Chen2015}. 

\begin{algorithm}
\caption{ Relay Selection and Time Minimization Algorithm (RSTMA)}\label{alg:MTNSKR}
\begin{algorithmic}[1]
\small \STATE $k_i=\argmax_{j\in\{0,\ldots,K\}}\min \left(g_{S_i}^{R_j}h_{AP}^{S_i}, g_{R_j}^{AP}h_{AP}^{R_j}\right), \; \forall i=1,\ldots,N$ 
\small \STATE $\tau \leftarrow$  POWMU ($S_i, R_{k_i}$),   $i=1,\ldots,N$ 
\small \STATE $\tau'=\tau$ 
\small \WHILE { $\tau' \leq \tau$}
\small \STATE $G_j= \{S_i: k_i=j, i=1\ldots,N\},  j=0,1,\ldots,K$
\small \STATE  $J =\text{sort}_j (|G_j|>1,'descend')$
\small \IF {$J \neq \emptyset$}
\small \FOR {$j\in J $}
\small \FOR { $S_i\in G_j$}

	\small \STATE $W =\textbf{sort}_j ( \min \left(g_{S_i}^{R_j}h_{AP}^{S_i}, g_{R_j}^{AP}h_{AP}^{R_j} \right),\; \forall j, j\neq k_i,'descend')$
	\small \FOR {$w \in W$}
		\small \STATE  $k_i \leftarrow w$ 
		\small \STATE $\tau' \leftarrow$  POWMU ($S_i, R_{k_i}$),   $i=1,\ldots,N$ 
		\small \IF {$\tau' < \tau$}
			\small \STATE $\tau = \tau'$
			\small \STATE \textbf{break}
		\small \ENDIF
	\small \ENDFOR
\small \ENDFOR
\small \ENDFOR
\small \ENDIF
\small \ENDWHILE
\end{algorithmic}
\end{algorithm}

The detailed procedure for  RSTMA is given in Algorithm~\ref{alg:MTNSKR}. Initially, we match each ${S_i}$ with $R_{k_i}$ by using the criterion given in Eq.~(\ref{eq:relaySelectCriteria}) (Line 1). POWMU finds the total schedule length $\tau$ for the given relay selection (Line 2). $\tau'$ is used to store a new schedule length and initially is set to $\tau$ (Line 3). The following procedure continues while there is an improvement in schedule length (Line 4). The sources assigned to relay $R_j$ are grouped under $G_j$ for all $j$ (Line 5). The relay with the greatest number of the sources is most probably the one requiring more EH time. Thus, starting with the ones belonging to the group of such relay, the sources are reassigned to the other relays one by one (Lines 6-9). As the aim is to lower the burden of the relay assisting multiple sources, the relays assisting single source is not in the interest of the reassignment procedure (Line 6). Since the relay $R_j$ with higher  $\min \left(g_{S_i}^{R_j}h_{AP}^{S_i}, g_{R_j}^{AP}h_{AP}^{R_j} \right)$ is expected to provide lower schedule length based on Eq.~(\ref{eq:relaySelectCriteria}),  a source is first reassigned to such a relay. If this change does not bring an improvement, other relay options are tried (Lines 10-13). Whenever an improvement is observed, $\tau$ is updated by $\tau'$; the cycle is broken for the sake of time complexity; the algorithm searches for another source-relay reassignment (Lines 14-16). 



\vspace*{-10mm}

\section{Performance Evaluation} \label{section:simRes}
In this section, the performance of the proposed algorithms,  BBA, OBH, RPH, and RSTMA, is evaluated via numerical simulations and compared with benchmark harvest-then-cooperate protocol (HTC) from \cite{Chen2015}, and the initial solution of RSTMA, denoted by OR+POWMU, as it uses OR criteria for relay selection and POWMU for schedule length calculation. HTC divides unit time block $T$ into $\rho T$ for EH and $(1-\rho)T$ for IT. It considers single source network and equal length IT blocks $(1-\rho)T/2$ for source-to-relay and relay-to-AP communication.  We adapt HTC scheme to multiple source network by allocating $(1-\rho)T/2N$ time for each  source-to-relay and relay-to-AP transmission. The unit time block is assumed as 1 ms, i.e., $T=1 ms$. $\rho$ is set to 0.8, which is the best value for schedule length minimization according to our simulations. We additionally include the enforcement for obeying $P^{max}$ limitation. For relay selection, OR criteria, corresponding to Eq.~(\ref{eq:relaySelectCriteria}), is applied.


Simulations are conducted over 1000 independent random network realizations in MATLAB. Sources are uniformly distributed in a circular quadrant between radius $3-4$ m. The AP is located at the center of the circle. Relays are deterministically positioned between the sources and the AP at distance $2$ m from the AP. $5$-user $2$-relay WPCCN is considered unless otherwise stated, as different sized networks behave similarly. The WPCCN is simulated at $915$ MHz central frequency with $1$ MHz bandwidth.

The channel attenuation is modeled as follows. Large scale statistics are formulated as $ PL(d) = PL(d_0) - 10\upsilon \log_{10}(d/d_0) + Z$, where $PL(d_0)$ is the free space path loss at unit distance $d_0$ in dB, $d$ is the distance between the transmitter and receiver, $PL(d)$ is the path loss at distance $d$ in dB, $\upsilon$ is the path-loss exponent, and $Z$ is a zero-mean Gaussian random variable with standard deviation $\sigma_Z$. For small scale statistics, Rayleigh fading model is used with scale parameter $\Omega$ set to the mean power level determined by $PL(d)$ \cite{Morsi2015}. The simulation parameters are chosen as $N_0 = -90dBm$ , $\zeta =50\%$ \cite{Le2008}, $PL(d_0)=31.67$, $\sigma_Z= 2dB^2$, and $\upsilon= 2$. The transmit power of the AP is chosen as $4$ Watts, which is the maximum allowed power at $915$ MHz band according to FCC regulations \cite{Lu2015}. All users are assumed to transmit equal amount of data, $50$ bits, to the AP.  

\begin{figure}[!ht]%
\centering
\subfloat[]{\label{fig:time_vs_N}
\includegraphics[width=.5\columnwidth]{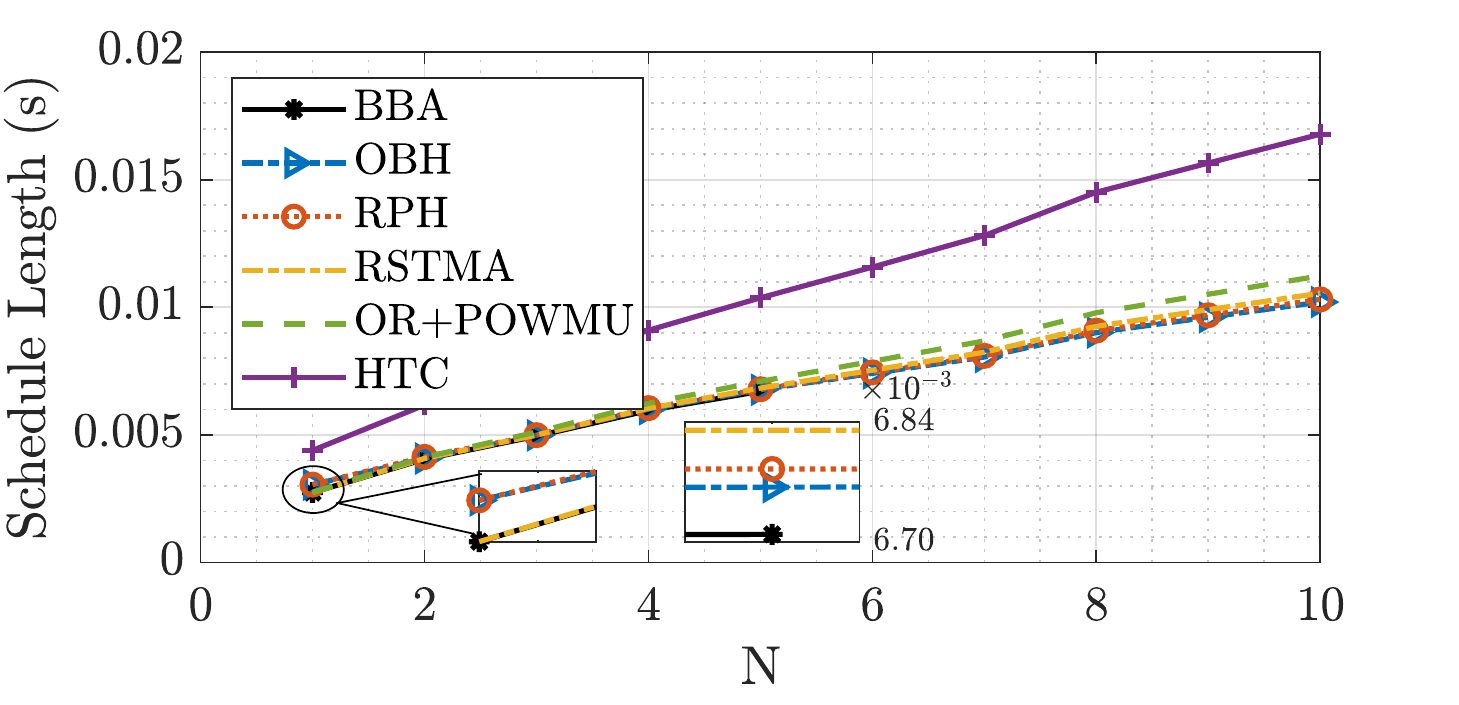}} 
\subfloat[]{ \label{fig:time_vs_K}
\includegraphics[width=.5\columnwidth]{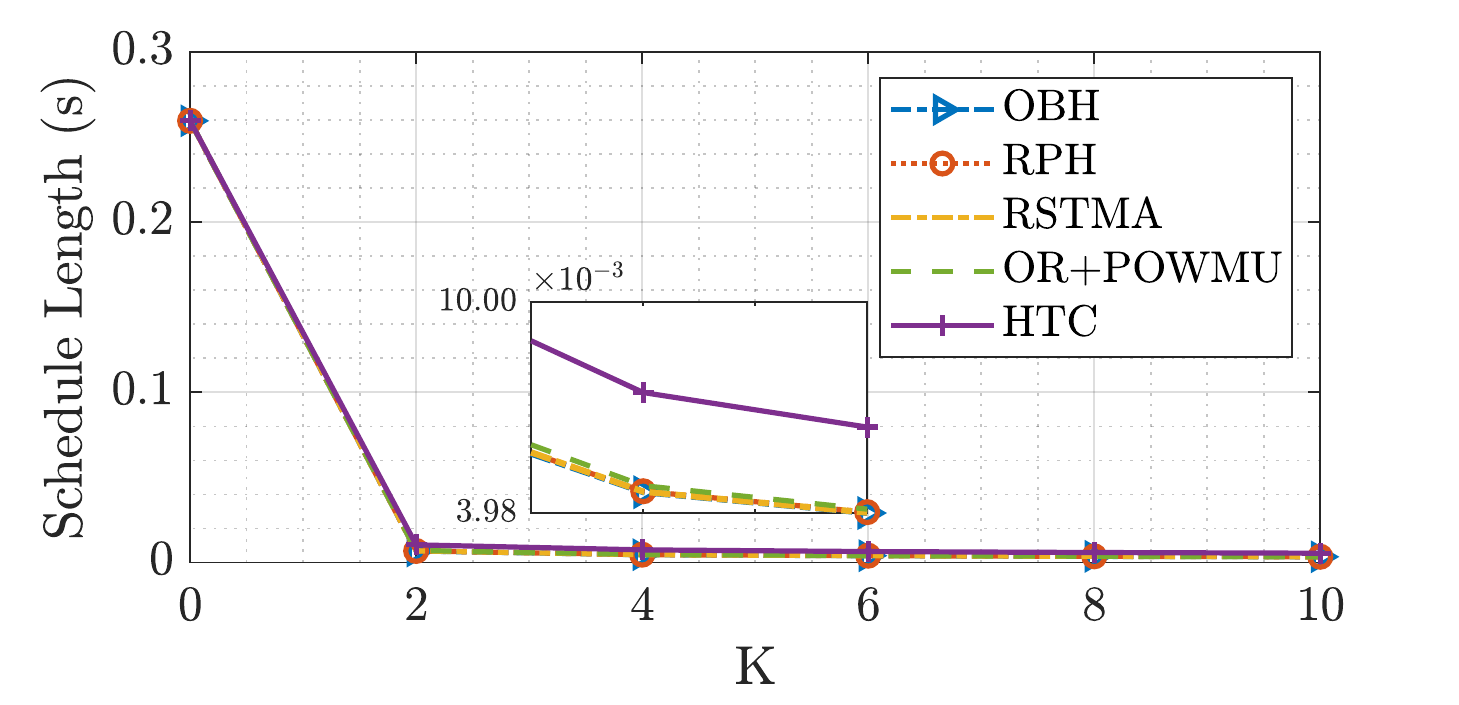}} \\
\subfloat[]{\label{fig:time_vs_pmax}
\includegraphics[width=.5\columnwidth]{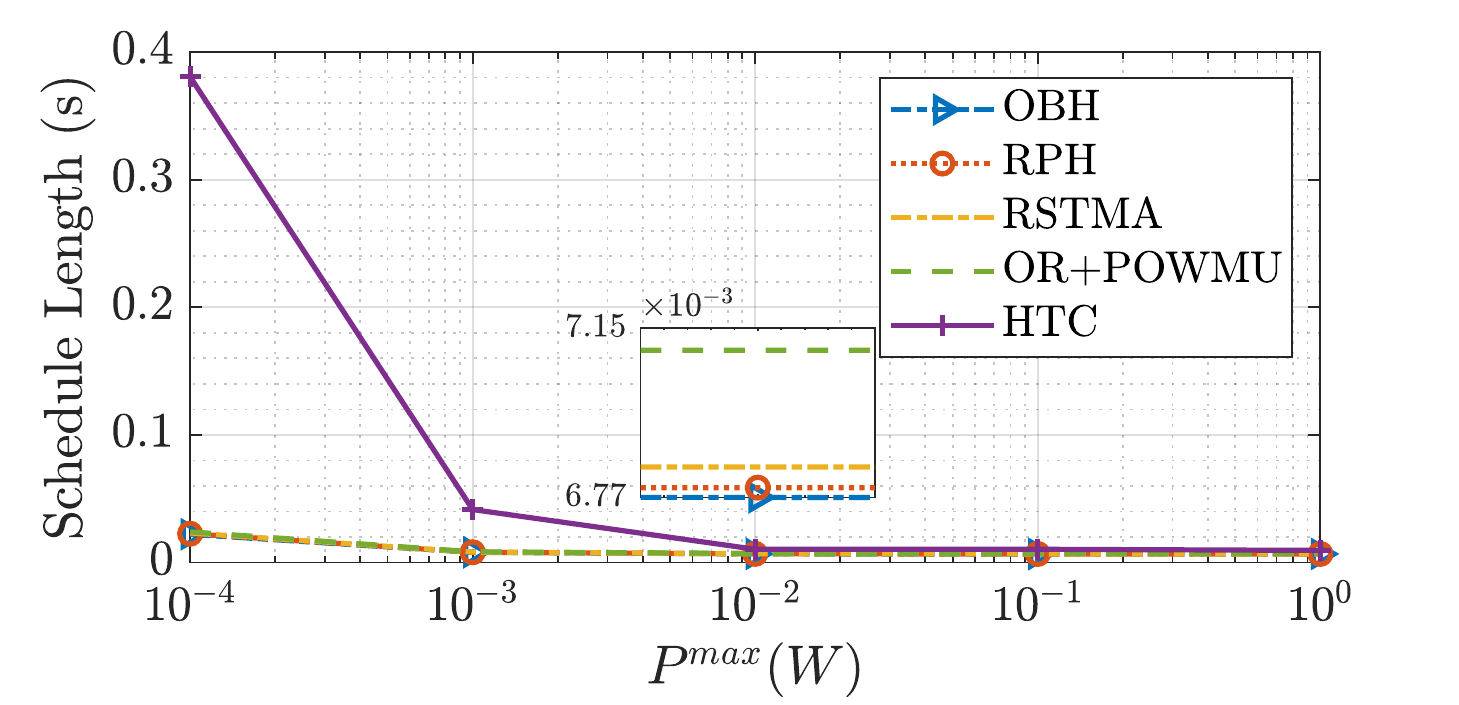}}%
\subfloat[]{\label{fig:relay_pos}
\includegraphics[width=.5\columnwidth]{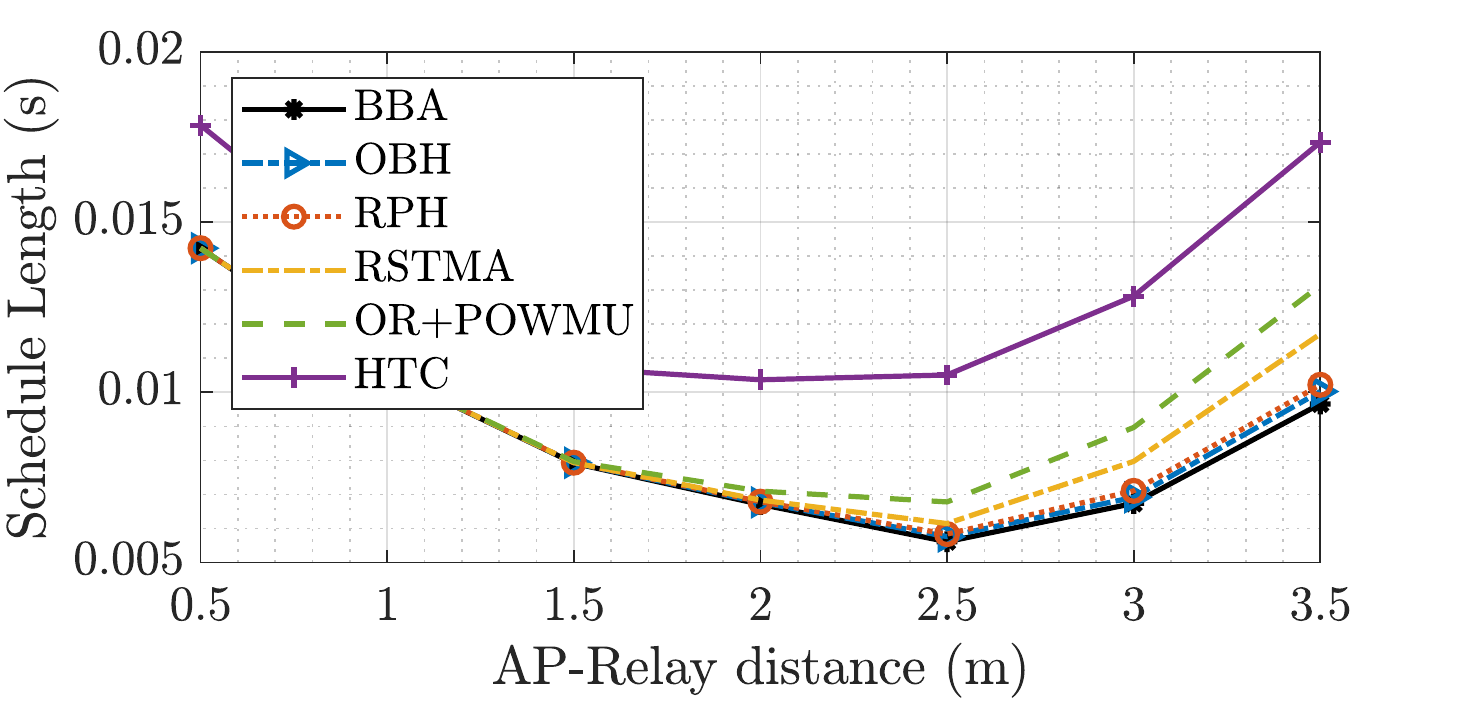}}\\%
\caption{Schedule length vs. (a) number of sources, (b) number of relays, (c) maximum power limit of the sources and relays, and (d) relay position.}
\label{fig:sim_results}
\end{figure}

Figs.~(\ref{fig:time_vs_N})-(\ref{fig:relay_pos}) depict the schedule length of the network obtained by different algorithms for varying number of users, number of relays, maximum UL transmit power limit ($P^{max}$), and relay position, respectively. BBA is only simulated for the first and fourth cases up to $5$ sources due to its exponential time complexity. The proposed algorithms outperform the benchmark approach HTC; whereas the heuristic approaches, OBH, RPH, and RSTMA, perform very closely to the optimum algorithm BBA. The difference in the schedule length computed by the proposed heuristics results from the assignment of the relays, as POWMU obtains the optimum schedule length for given relay selection. OBH performs best, RPH is the second best, which is followed by RSTMA.  Comparison of RSTMA with OR+POWMU indicates the significance of the further improvement by RSTMA from its initial point determined by OR+POWMU, especially for higher number of users. This observation highlights the fact that the individual best relay selection is not necessarily the best for the entire network. The figures are further analyzed as follows.

In Fig.~(\ref{fig:time_vs_N}), the schedule length increases with the increasing number of users, as expected. For $N=5$, BBA provides $35\%$ improvement in schedule length compared to HTC; whereas the gap between BBA and RSTMA is $1.86\%$, the gap between BBA and RPH is $1.18\%$, and the gap between BBA and OBH is only $0.85\%$. Note that for smaller number of users, RSTMA provides lower schedule length than RPH.
 
In Fig.~(\ref{fig:time_vs_K}), the schedule length decreases with the increasing number of relays, which indicates the benefit of the relays. Using $2$ relays reduces the schedule length of the network by $97\%$; whereas using $10$ relays decreases it up to $98.6\%$. This demonstrates diminishing returns as the number of relays increases.


In Fig.~(\ref{fig:time_vs_pmax}), the gap between the HTC and proposed algorithms is higher at lower $P^{max}$ levels as HTC does not consider transmit power limitation, e.g., $88\%$ gap for $P^{max}= 10^{-4}  W$. The gap closes as $P^{max}$ increases, however, there is still $20\%$ gap between HTC and OBH for $P^{max}= 1 W$.

Fig.~(\ref{fig:relay_pos}) depicts the schedule length of the network for different relay positions, where the AP and sources stay at the same positions. The relays bring the highest improvement when they are located at $2.5$ m, closer to the sources than the AP. Between $0.5-2.5$ m,  all the proposed heuristics performs very close to BBA. After $2$ m, the performance of RSTMA degrades.
 
Fig.~(\ref{fig:runtime_vs_N}) shows the run-time performance of the proposed algorithms. BBA has exponential complexity in practical simulations as well. The runtime of OBH rapidly increases with the increasing number of users as it has to solve the relaxed problem many times and the complexity of this solution depends on the number of users.  The runtime of RPH linearly increases as it solves the relaxed problem just once. For $N=5$, RSTMA performs $99\%$ and $97\%$ faster than OBH and RPH, respectively. RSTMA has higher runtime than OR+POWMU, as it improves the OR criteria with further iterations depending on the number of sources. 

\begin{figure}[!ht]
    \centering
    \includegraphics[width=.5\columnwidth]{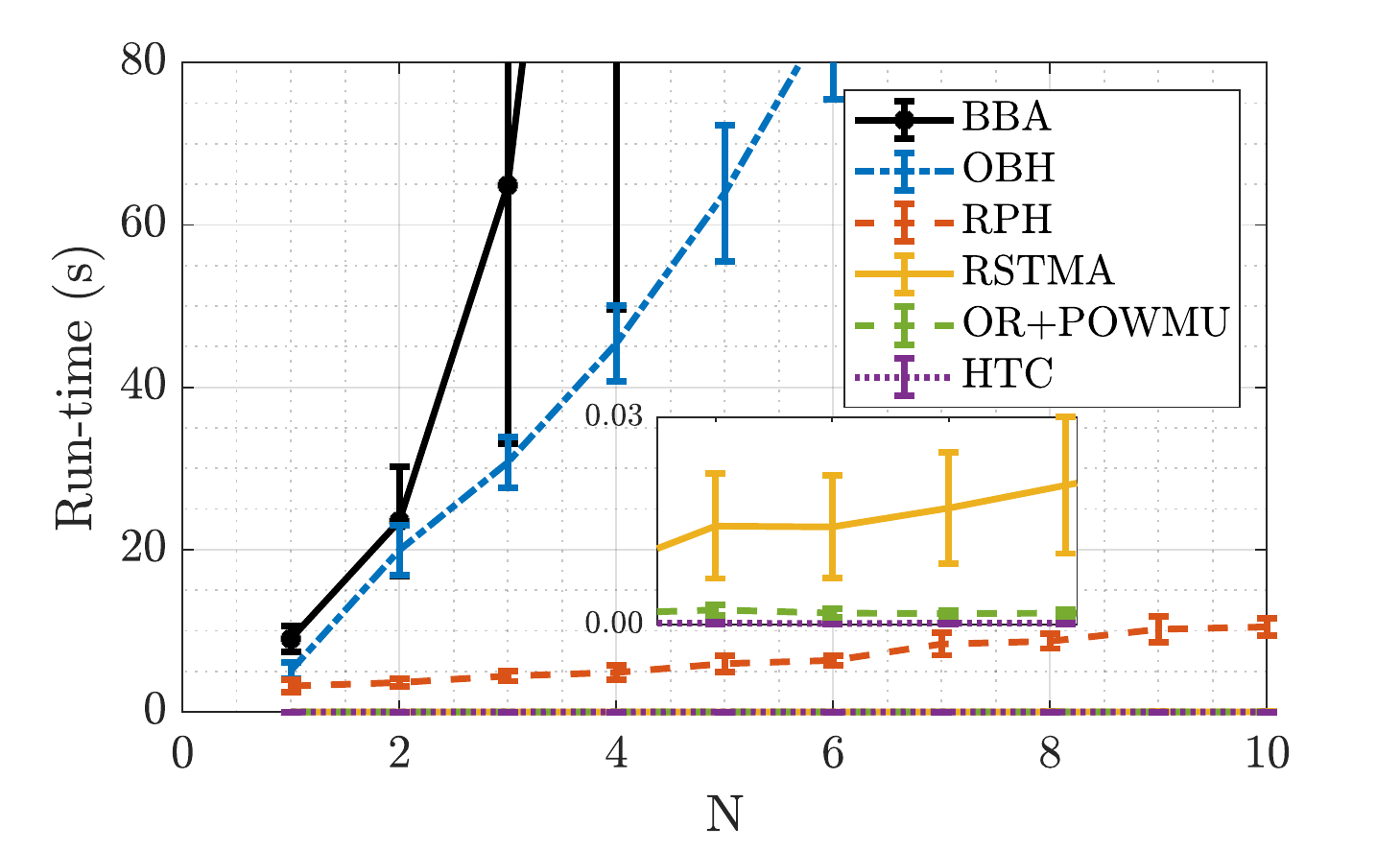}
    \caption{Run-time of the algorithms vs. number of users.}
    \label{fig:runtime_vs_N}
\end{figure}


\vspace*{-10mm}

\section{Conclusion} \label{section:conc}
We formulate a novel joint relay selection, scheduling and power control problem for multiple source multiple relay WPCCN. The problem aims to minimize the total schedule duration of energy harvesting and information transmission, while satisfying data, energy causality, and maximum UL transmit power constraints. The problem is non-convex MINLP and proven to be NP-hard. Given the relay selection, we first formulate a scheduling and power control problem, and then propose an efficient optimal and a faster sub-optimal algorithm for the solution. To determine relay assignments, we first adapt an exponential-time branch-and-bound (BB) based algorithm for the optimal solution of the problem. Then, two BB-based heuristic approaches are proposed together with a lower complexity relay criterion based algorithm. We demonstrate via numerical simulations that our proposed approaches exhibit better performance than the conventional harvest-then-cooperate approaches with up to $88\%$ lower schedule length at various network settings, especially at lower UL transmit power limits.

\vspace*{-5mm}

\bibliographystyle{IEEEtran}
\bibliography{relSel_onecolumn}

\begin{thebibliography}{10}
\providecommand{\url}[1]{#1}
\csname url@samestyle\endcsname
\providecommand{\newblock}{\relax}
\providecommand{\bibinfo}[2]{#2}
\providecommand{\BIBentrySTDinterwordspacing}{\spaceskip=0pt\relax}
\providecommand{\BIBentryALTinterwordstretchfactor}{4}
\providecommand{\BIBentryALTinterwordspacing}{\spaceskip=\fontdimen2\font plus
\BIBentryALTinterwordstretchfactor\fontdimen3\font minus
  \fontdimen4\font\relax}
\providecommand{\BIBforeignlanguage}[2]{{%
\expandafter\ifx\csname l@#1\endcsname\relax
\typeout{** WARNING: IEEEtran.bst: No hyphenation pattern has been}%
\typeout{** loaded for the language `#1'. Using the pattern for}%
\typeout{** the default language instead.}%
\else
\language=\csname l@#1\endcsname
\fi
#2}}
\providecommand{\BIBdecl}{\relax}
\BIBdecl

\bibitem{Salik19}
E.~D. Salik, A.~G. Onalan, and S.~{Coleri Ergen}, ``Minimum length scheduling
  for power constrained harvest-then-transmit communication networks,'' in
  \emph{IEEE PIMRC}, Sept. 2019.

\bibitem{Lu2015}
X.~Lu, P.~Wang, D.~Niyato, D.~I. Kim, and Z.~Han, ``Wireless networks with {RF}
  energy harvesting: A contemporary survey,'' \emph{{IEEE} Communications
  Surveys {\&} Tutorials}, vol.~17, no.~2, pp. 757--789, 2015.

\bibitem{Perera2018}
T.~D.~P. Perera, D.~N.~K. Jayakody, S.~K. Sharma, S.~Chatzinotas, and J.~Li,
  ``Simultaneous wireless information and power transfer ({SWIPT}): Recent
  advances and future challenges,'' \emph{{IEEE} Communications Surveys {\&}
  Tutorials}, vol.~20, no.~1, pp. 264--302, 2018.

\bibitem{Kang2015}
X.~Kang, C.~K. Ho, and S.~Sun, ``Full-duplex wireless-powered communication
  network with energy causality,'' \emph{{IEEE} Transactions on Wireless
  Communications}, vol.~14, no.~10, pp. 5539--5551, Oct. 2015.

\bibitem{sinem2019}
I.~Pehlivan and S.~{Coleri Ergen}, ``Scheduling of energy harvesting for mimo
  wireless powered communication networks,'' \emph{IEEE Communications
  Letters}, vol.~23, no.~1, pp. 152--155, Jan. 2019.

\bibitem{Chen2012}
J.~Chen, A.~Massouri, L.~Clavier, C.~Loyez, N.~Rolland, and P.~Rolland, ``Relay
  characteristic impact on energy consumption in heterogeneous sensor
  network,'' \emph{{AEU} - International Journal of Electronics and
  Communications}, vol.~66, no.~6, pp. 495--501, Jun. 2012.

\bibitem{Nosratinia2004}
A.~Nosratinia, T.~Hunter, and A.~Hedayat, ``Cooperative communication in
  wireless networks,'' \emph{{IEEE} Communications Magazine}, vol.~42, no.~10,
  pp. 74--80, Oct. 2004.

\bibitem{Mishra2017}
D.~Mishra, S.~De, and D.~Krishnaswamy, ``Dilemma at {RF} energy harvesting
  relay: Downlink energy relaying or uplink information transfer?''
  \emph{{IEEE} Transactions on Wireless Communications}, vol.~16, no.~8, pp.
  4939--4955, Aug. 2017.

\bibitem{Ammar2018}
A.~Ammar and D.~Reynolds, ``Energy harvesting networks: Energy versus data
  cooperation,'' \emph{{IEEE} Communications Letters}, vol.~22, no.~10, pp.
  2128--2131, Oct. 2018.

\bibitem{Nasir2015}
A.~A. Nasir, X.~Zhou, S.~Durrani, and R.~A. Kennedy, ``Wireless-powered relays
  in cooperative communications: Time-switching relaying protocols and
  throughput analysis,'' \emph{{IEEE} Transactions on Communications}, vol.~63,
  no.~5, pp. 1607--1622, May 2015.

\bibitem{Gu2015}
Y.~Gu and S.~Aissa, ``{RF}-based energy harvesting in decode-and-forward
  relaying systems: Ergodic and outage capacities,'' \emph{{IEEE} Transactions
  on Wireless Communications}, vol.~14, no.~11, pp. 6425--6434, Nov. 2015.

\bibitem{Ju2015}
M.~Ju, K.-M. Kang, K.-S. Hwang, and C.~Jeong, ``Maximum transmission rate of
  {PSR}/{TSR} protocols in wireless energy harvesting {DF}-based relay
  networks,'' \emph{{IEEE} Journal on Selected Areas in Communications},
  vol.~33, no.~12, pp. 2701--2717, Dec. 2015.

\bibitem{Chen2015}
H.~Chen, Y.~Li, J.~L. Rebelatto, B.~F. Uchoa-Filho, and B.~Vucetic,
  ``Harvest-then-cooperate: Wireless-powered cooperative communications,''
  \emph{{IEEE} Transactions on Signal Processing}, vol.~63, no.~7, pp.
  1700--1711, Apr. 2015.

\bibitem{Gu2015a}
Y.~Gu, H.~Chen, Y.~Li, and B.~Vucetic, ``An adaptive transmission protocol for
  wireless-powered cooperative communications,'' in \emph{2015 {IEEE}
  International Conference on Communications ({ICC})}.\hskip 1em plus 0.5em
  minus 0.4em\relax {IEEE}, Jun 2015.

\bibitem{Li2016}
X.~Li, Q.~Tang, and C.~Sun, ``The impact of node position on outage performance
  of {RF} energy powered wireless sensor communication links in overlaid
  deployment scenario,'' \emph{Journal of Network and Computer Applications},
  vol.~73, pp. 1--11, Sep. 2016.

\bibitem{Krikidis2014}
I.~Krikidis, ``Simultaneous information and energy transfer in large-scale
  networks with/without relaying,'' \emph{{IEEE} Transactions on
  Communications}, vol.~62, no.~3, pp. 900--912, Mar. 2014.

\bibitem{Gu2018}
Y.~Gu, H.~Chen, Y.~Li, Y.-C. Liang, and B.~Vucetic, ``Distributed multi-relay
  selection in accumulate-then-forward energy harvesting relay networks,''
  \emph{{IEEE} Transactions on Green Communications and Networking}, vol.~2,
  no.~1, pp. 74--86, Mar. 2018.

\bibitem{Wang2018}
F.~Wang, S.~Guo, Y.~Yang, and B.~Xiao, ``Relay selection and power allocation
  for cooperative communication networks with energy harvesting,'' \emph{{IEEE}
  Systems Journal}, vol.~12, no.~1, pp. 735--746, Mar. 2018.

\bibitem{Zhao2016}
Y.~Zhao, Q.~Li, L.~Huang, S.~Feng, T.~Han, and J.~Zhang, ``Wireless information
  and power transfer on cooperative multi-path relay channels,'' in \emph{2016
  IEEE/CIC International Conference on Communications in China (ICCC)}, July
  2016, pp. 1--6.

\bibitem{Sui2018}
D.~Sui, F.~Hu, W.~Zhou, M.~Shao, and M.~Chen, ``Relay selection for radio
  frequency energy-harvesting wireless body area network with buffer,''
  \emph{{IEEE} Internet of Things Journal}, vol.~5, no.~2, pp. 1100--1107, Apr.
  2018.

\bibitem{Nasir2018}
H.~Nasir, N.~Javaid, M.~Imran, M.~Shoaib, and M.~Anwar, ``Simultaneous wireless
  information and power transfer for buffer-aided cooperative relaying
  systems,'' in \emph{2018 14th International Wireless Communications {\&}
  Mobile Computing Conference ({IWCMC})}.\hskip 1em plus 0.5em minus
  0.4em\relax {IEEE}, Jun. 2018.

\bibitem{Sadi2014}
Y.~Sadi and S.~{Coleri Ergen}, ``Minimum length scheduling with packet traffic
  demands in wireless ad hoc networks,'' \emph{IEEE Transactions on Wireless
  Communications}, vol.~13, no.~7, pp. 3738--3751, July 2014.

\bibitem{Hou2012}
I.-H. Hou and P.~R. Kumar, ``Real-time communication over unreliable wireless
  links: a theory and its applications,'' \emph{{IEEE} Wireless
  Communications}, vol.~19, no.~1, pp. 48--59, Feb. 2012.

\bibitem{Ding2014}
Z.~Ding, I.~Krikidis, B.~S. Sharif, and H.~V. Poor, ``Wireless information and
  power transfer in cooperative networks with spatially random relays,''
  \emph{{IEEE} Transactions on Wireless Communications}, vol.~13, no.~8, pp.
  4440--4453, Aug. 2014.

\bibitem{Ju2014}
H.~Ju and R.~Zhang, ``Throughput maximization in wireless powered communication
  networks,'' \emph{{IEEE} Transactions on Wireless Communications}, vol.~13,
  no.~1, pp. 418--428, Jan. 2014.

\bibitem{Zhang2016}
D.~Zhang, Z.~Chen, H.~Zhou, L.~Chen, and X.~Shen, ``Energy-balanced cooperative
  transmission based on relay selection and power control in energy harvesting
  wireless sensor network,'' \emph{Computer Networks}, vol. 104, pp. 189 --
  197, 2016.

\bibitem{OpResBook}
F.~S. Hillier and G.~J. Lieberman, \emph{Introduction to Operations Research,
  9th Ed.}\hskip 1em plus 0.5em minus 0.4em\relax New York: McGrawHill, 2010.

\bibitem{Chiang2007}
M.~{Chiang}, S.~H. {Low}, A.~R. {Calderbank}, and J.~C. {Doyle}, ``Layering as
  optimization decomposition: A mathematical theory of network architectures,''
  \emph{Proceedings of the IEEE}, vol.~95, no.~1, pp. 255--312, Jan. 2007.

\bibitem{convexVectorAnalysis}
C.~Zalinescu, \emph{Convex Analysis in General Vector Spaces}.\hskip 1em plus
  0.5em minus 0.4em\relax River Edge, N.J.: World Scientific, 2002.

\bibitem{Trench2013}
\BIBentryALTinterwordspacing
W.~F. Trench, \emph{Introduction to Real Analysis}.\hskip 1em plus 0.5em minus
  0.4em\relax Faculty Authored and Edited Books \& CDs.7, 2013. [Online].
  Available: \url{https://digitalcommons.trinity.edu/mono/7}
\BIBentrySTDinterwordspacing

\bibitem{PietroBelotti2012}
P.~Belotti, C.~Kirches, S.~Leyffer, J.~Linderoth, J.~Luedtke, and A.~Mahajan,
  ``Mixed-integer nonlinear optimization,'' Argonne National Laboratory, Tech.
  Rep., 2012.

\bibitem{cvx}
M.~Grant and S.~Boyd, ``{CVX}: Matlab software for disciplined convex
  programming, version 2.1,'' \url{http://cvxr.com/cvx}, Mar. 2014.

\bibitem{Morsi2015}
R.~{Morsi}, D.~S. {Michalopoulos}, and R.~{Schober}, ``Multiuser scheduling
  schemes for simultaneous wireless information and power transfer over fading
  channels,'' \emph{IEEE Transactions on Wireless Communications}, vol.~14,
  no.~4, pp. 1967--1982, Apr. 2015.

\bibitem{Le2008}
T.~Le, K.~Mayaram, and T.~Fiez, ``Efficient far-field radio frequency energy
  harvesting for passively powered sensor networks,'' \emph{{IEEE} Journal of
  Solid-State Circuits}, vol.~43, no.~5, pp. 1287--1302, May 2008.

\end{thebibliography}
\end{document}